\documentclass[11pt]{article}%
\usepackage{amssymb}
\usepackage{amsfonts}
\usepackage{amsmath}
\usepackage{amsbsy}
\usepackage{graphicx}
\usepackage{enumerate}
\usepackage{geometry}
\usepackage{amssymb}
\usepackage{amsfonts}
\usepackage{amsmath}
\usepackage{amsbsy}
\usepackage{graphicx}
\usepackage{enumerate}
\usepackage{shuffle}

\usepackage{tikz}
\usetikzlibrary{arrows,shapes.geometric,positioning}

\usepackage[hidelinks]{hyperref}

\providecommand{\U}[1]{\protect\rule{.1in}{.1in}}

\allowdisplaybreaks
\newtheorem{theorem}{Theorem}

\newtheorem{corollary}[theorem]{Corollary}

\newtheorem{definition}[theorem]{Definition}
\newtheorem{example}[theorem]{Example}

\newtheorem{proposition}[theorem]{Proposition}

\newenvironment{proof}[1][Proof]{\noindent\textbf{#1.} }{\ \rule{0.5em}{0.5em}}
\definecolor{darkorange}{rgb}{1.0, 0.55, 0.0}

\allowdisplaybreaks
\newcommand{\B}{{\mathcal{B}}}
\newcommand{\wB}{{\mathcal{wB}}}
\newcommand{\e}{{\tilde{\varphi}}}
\newcommand{\x}{{\tilde{\xi}}}
\newcommand{\ps}{{\tilde{\psi}}}

\newcommand{\blue}[1]{{{\textcolor{blue}{#1}}}}

\newcommand{\hide}[1]{}
\providecommand{\keywords}[1]{\textbf{Keywords:} #1}
\begin{document}

\title{Modelling architectures of parametric weighted \\ component-based systems}
\author{Maria Pittou$^{a,}$\thanks{\protect\includegraphics[height=0.3cm]{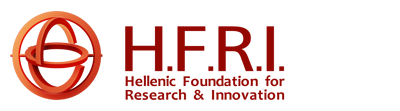}The research work was supported by the Hellenic Foundation for Research and Innovation (HFRI) under the HFRI PhD Fellowship grant (Fellowship Number: 1471).} \ and George Rahonis$^{b}$ \\Department of Mathematics\\Aristotle University of Thessaloniki\\54124 Thessaloniki, Greece\\$^{a}$mpittou@math.auth.gr, $^{b}$grahonis@math.auth.gr}
\date{}
\maketitle

\begin{abstract}The design of complex software systems usually lies
 in multiple coordinating components with an unknown number of instances. 
For such systems a main challenge is modelling efficiently their architecture that determines
 the topology and the interaction principles among the components.
To achieve well-founded design there is need to address the quantitative aspects of software architectures. 
In this paper we study the modelling problem of software architectures applied on parametric weighted component-based systems
where the parameter is the number of instances of each component. For this, we introduce a weighted first-order extended interaction
logic over a commutative semiring in order to serve as a modelling language for parametric quantitative architectures. 
We prove that the equivalence problem of formulas of that logic is decidable in the class (of subsemirings) of skew fields.
Moreover, we show that our weighted logic can efficiently describe well-known
parametric architectures with quantitative characteristics.

\

\noindent \keywords{Architecture modelling. Parametric weighted component-based systems. 
Weighted first-order extended interaction logic.}
\end{abstract}

\section{Introduction}
One of the key challenges in software systems' design is the specification of their architecture. Software architectures depict design principles applied to systems in order to characterize the communication rules and the underlying topology \cite{Ma:Co,Sh:Co}.  In such a setting systems' modelling is mainly component-based where the development lies in the coordination of multiple system components \cite{Bl:Al,Tr:Co}. Therefore, as the systems grow large and complex, architectures become important in the design process since they enforce systems to meet most of their requirements by construction. On the other hand, well-founded design involves not only the qualitative features but also the non-functional aspects of the systems \cite{Kw:Qu,No:Qu}, and hence of their architectures \cite{Pa:On, Si:Ri}. Qualitative modelling cannot capture the timing constraints, available resources, energy consumption, etc., for executing the interactions of a system's architecture. Such optimization and performance aspects require the study of architectures in the quantitative setting. 

In this paper we study architectures by a weighted first-order interaction logic. We consider architectures applied on a wide class of component-based systems, namely parametric weighted systems. Parametric systems consist of 
a finite number of component types with an
unknown number of instances \cite{Am:Pa,Bl:De}. Such systems are found in several applications including cyber-physical systems \cite{Li:St,Ta:In}, as well as distributed algorithms and communication protocols \cite{Ab:Pa,Bl:De,De:Pa,Es:Po}. Quantitative analysis is important for parametric architectures for handling communication issues like `the battery waste for a particular interaction is' 
in a system
with restricted power autonomy like sensor networks or `the required time for completing the interactions among specific components is' 
in time-critical network protocols. 

According to our best knowledge a logical directed description of parametric quantitative architectures, using weighted  first-order configuration logic, has been only considered in \cite{Pa:On} (cf. Section 2). The current paper extends the results of our recent work for modelling parametric architectures \cite{Pi:Ar}, in the quantitative setup. Specifically, in this work we model parametric systems by a classical weighted model, namely weighted labelled transition system. The communication is achieved by the labels, called ports, and is
defined by interactions, i.e., sets of ports. Each port is associated with a weight that represents the `cost' of its participation in an interaction. The weights of the ports range over the common algebraic structure of a commutative semiring $K$. We formalize parametric quantitative  architectures by weighted logic formulas over $K$. Similarly to the work of \cite{Pi:Ar}, our weighted logic has the additional attribute of respecting the execution order of the interactions as imposed by each architecture, a main feature of several important architectures found in applications. In particular, the contributions of this work are the following. 

(1) We introduce \emph{weighted Extended Propositional Interaction Logic} (wEPIL for short) over a finite set of ports,
and a commutative semiring $K$ for representing the weights. wEPIL extends \emph{weighted Propositional Interaction Logic} (wPIL for short) from \cite{Pa:On,Pa:We} with two new weighted operators, namely the weighted concatenation operator $\odot$ and weighted shuffle operator $ \varpi$. Intuitively these two operators allow us to encode the weights of consecutive and interleaving interactions in weighted component-based systems, respectively. We interpret wEPIL formulas as series defined over finite words and $K$. The letters of the words are interactions over the given set of ports. Clearly, the semantics of wEPIL formulas differs from the ones of wPIL formulas. The latter is interpreted over series from sets of interactions, whereas the first one over series of words of interactions, both with values in $K$. 

(2) We apply wEPIL formulas for modelling the architectures of weighted component-based systems with ordered interactions, and  specifically, we consider the architectures \emph{Blackboard} \cite{Co:Bl}, \emph{Request/Response} \cite{Da:Se}, and \emph{Publish/Subscribe} \cite{Eg:Pu}. For different instantiations of the semiring $K$ we derive alternative interpretations
for the resulting total cost of the allowed interactions, that corresponds to some quantitative characteristic.

(3) We introduce the first-order level of wEPIL, namely \emph{weighted First-Order Extended Interaction Logic} (wFOEIL for short). 
The syntax of wFOEIL is equipped with the syntax of wEPIL, the weighted existential and universal quantifiers (similar weighted quantifiers introduced firstly for weighted MSO logics in \cite{Dr:We}), and four new weighted quantifiers, namely the weighted existential and universal concatenation and shuffle quantifiers. The weighted existential and universal concatenation (resp. shuffle) quantifiers serve to compute the weight of the partial and whole participation of the instances of a weighted component type in sequentially (resp. interleaving) executed interactions. For the semantics of wFOEIL we consider triples consisting of 
a mapping defining the number of component instances in the parametric weighted system,
an assignment that attributes a
unique identifier to ports of each component instance, and a finite word of
interactions. Then, we interpret wFOEIL formulas as series from triples of the previous form to elements in $K$.

(4) We show that wFOEIL serves sufficiently for modelling interactions of parametric weighted component-based systems. Specifically, we provide examples of wFOEIL formulas for modelling the quantitative properties of concrete parametric architectures including \emph{Master/Slave} \cite{Ma:Co}, \emph{Star} \cite{Ma:Co}, \emph{Repository} \cite{Cl:Do}, and \emph{Pipes/Filters} \cite{Ga:An} where interactions are executed arbitrarily, as well as for the \emph{Blackboard}, \emph{Request/Response}, and \emph{Publish/Subscribe} architectures. 

(5) We prove the decidability of equivalence of wFOEIL formulas in doubly-exponential time provided that the weight structure is (a subesemiring of) a skew field. For this, we follow a methodology similar to the one considered in \cite{Pi:Ar}, and establish a doubly exponential translation of wFOEIL formulas to weighted automata. Then, we apply well-known decidability and complexity results for weighted automata over (subsemirings of) skew fields.

The structure of the paper is as follows. In Section $2$ we refer to related work and compare our results with the existing ones found in literature. In Section 3 we recall the notions of weighted component-based systems and weighted interactions. Then, in Section 4 we define the syntax and semantics of wEPIL and present examples of component-based systems whose architectures are defined by wEPIL formulas. In Section $5$ we introduce the syntax and semantics of our wFOEIL and provide examples of wFOEIL sentences for modelling specific architectures of parametric weighted component-based systems. Section $5$ studies the decidability results for wFOEIL sentences. Finally, in Conclusion, we discuss future work directions.

\section{Related work}
Existing work has recently investigated the architecture modelling problem of parametric systems mainly in the qualitative setting. 
In particular, in \cite{Ma:Co} the authors introduced a Propositional Configuration Logic (PCL for short) as a
modelling language for the description of architectures. They considered also first- and second-order
configuration logic applied for modelling parametric architectures (called styles of architectures in that
paper). PCL which is interpreted over sets of interactions has a nice property, namely for
every PCL formula one can efficiently construct an equivalent one in a special form called full
normal form. This implies the decidability of equivalence of PCL formulas in an automated
way using the Maude language. 

The weighted version of PCL over commutative semirings was studied in \cite{Pa:On,Pa:We}.  Soundness was proved for that logic and the authors obtained the decidability for the equivalence problem of its formulas.
Weighted PCL was applied for modelling the quantitative aspects of several common architectures. 
In \cite{Ka:We} the authors extended the work of \cite{Pa:On,Pa:We} and studied weighted PCL over a
product valuation monoid. That algebraic structure allows to compute the average cost as well as the maximum
cost among all costs occurring most frequently for executing the interactions within architectures.
The authors applied that logic to model
several weighted software architectures and proved that the equivalence problem of its formulas is decidable. In contrast to our setting, the work of \cite{Ka:We,Ma:Co,Pa:On, Pa:We} describes architectures with
PCL which though encodes no execution order of the interactions. Parametric weighted architectures were considered only in \cite{Pa:On} using the weighted first-order level of weighted PCL, namely weighted FOCL. Nevertheless, weighted FOCL considers no execution order of the interactions of parametric weighted architectures.

In \cite{Ko:Pa} a
first-order interaction logic (FOIL for short) was introduced to describe finitely many interactions, for parametric
systems modelled in BIP framework. FOIL was applied for modelling classical
architectures of parametric systems (such as Ring and Star) and considered in model checking results. 
In \cite{Bo:Ch} the authors introduced an alternative logic, namely monadic interaction logic (MIL for short)
to describe
rendezvous and broadcast communication of parametric component-based systems, and also presented a method for detecting deadlocks for these systems. Next, an interaction logic
with one successor (IL1S for short) was developed in \cite{Bo:St} as a modelling language for architectures of
parametric component-based systems. A method for checking deadlock freeness and mutual exclusion of parametric systems was also studied. 

Recently we introduced an extended propositional interaction logic (EPIL for short) and its
first-order level, namely first-order extended interaction logic (FOEIL for short) \cite{Pi:Ar}. FOEIL was applied for modelling well-known architectures of parametric component-based systems defined by labelled transition systems, and we proved the decidability of equivalence, satisfiability and validity of FOEIL sentences. In
comparison to logics of \cite{Bo:Ch,Bo:St,Ko:Pa,Ma:Co}, FOEIL not only described the permissible interactions characterizing each
architecture, but also encoded the execution order of implementing the interactions, an important feature of several architectures including Publish/Subscribe and Request/Response. 

Multiparty session types described efficiently communication protocols  and their interactions patterns (cf. for instance \cite{Ho:Mu,Hu:Fo}). A novel session type
system and an associated programming language were introduced in \cite{Ch:Pa} for the  description of multi-actor communication in parameterized protocols. The system modelled global types with exclusive and concurrent events as well as
their arbitrary reorderings using a shuffling operator, and used for static verification of asynchronous communication protocols. Though no architectures of systems were considered in \cite{Ch:Pa,Ho:Mu,Hu:Fo}. In \cite{De:Pa} the authors introduced a
type theory for multiparty sessions to globally specify parameterized
communication protocols whose processes carry data. They also developed decidable projection methods of parameterised global
types to local ones and proved type safety and deadlock freedom
for well-typed parameterized processes.
The framework of \cite{De:Pa} did not address the sequential and interleaving interactions of processes.

Some work for quantitative parametric systems was considered in \cite{An:Pa, Be:Pa, Es:Po, Fo:Th}. In
\cite{Es:Po} the authors studied population protocols, a specific class of parametric systems, modelled
by labelled transition systems with Markov chains semantics. Then, a decidability result was
obtained for the model checking problem of population protocols against linear-time specifi-
cations while undecidability was proved for the corresponding probabilistic properties. The work of \cite{Es:Po} does not consider though the systems' architecture in the design process. In \cite{Be:Pa, Fo:Th} the authors considered broadcast communication and cliques topology for networks
of many identical probabilistic timed processes, where the number of processes was a parameter.
Then, they investigated the decidability results of several qualitative parameterized verification problems
for probabilistic timed protocols.
In the subsequent work \cite{An:Pa} the authors extended broadcast protocols
and parametric timed automata and introduced a model of 
parametric timed broadcast
protocols with two different types of parameters,
namely the number of identical processes and the timing features. The
decidability of parametric decision problems were also studied for parametric timed broadcast
protocols in \cite{An:Pa}. In contrast to our framework, the topologies of the protocols studied in 
the above works, are not formalized by means of weighted logics. 
Moreover, we investigate the modelling problem of several more complicated quantitative parameteric
 architectures than those considered in \cite{An:Pa,Be:Pa, Fo:Th}.

The motivation to study parametric systems with quantitative features is also depicted in the recent work of \cite{Ho:Pr, Le:Fa}. 
In \cite{Le:Fa} the authors studied fair termination for parameterized probabilistic
concurrent systems modelled by Markov decision processes. They 
extended the symbolic framework 
of regular model checking for verifying parameterized
concurrent systems in the probabilistic setting. The authors developed a
fully-automatic method that was applied for distributed algorithms and systems in evolutionary biology. Moreover, the work of \cite{Ho:Pr} studied the equivalence problem for probabilistic parameterized systems using bisimulations.
For this, the first-order theory of universal
automatic structures was extended to a probabilistic setup for developing
a decidable automatic method 
for verifying probabilistic bisimulation for parameterized systems. Then, the framework was applied for
studying the anonymity property for the dining cryptographers and grades parameterized protocols.\hide{The recent works of \cite{Ho:Pr, Le:Fa} do not study the architecture modelling problem, though the indicate the need to study parametric systems with quantitative features.} Although the works of \cite{Ho:Pr, Le:Fa} apply their verification results on parametric protocols with ring or linear topologies, 
they do not focus on the investigation of a formal modelling framework for quantitative parametric architectures.

\section{Preliminaries}

\subsection{Notations}
For every natural number $n \geq 1$ we denote by $[n]$ the set $\{1, \ldots, n \}$. Hence, in the sequel, whenever we use the notation $[n]$ we always assume that $n \geq 1$. For every set $S$ we denote by $\mathcal{P}(S)$ the powerset of $S$. Let $A$ be an alphabet, i.e., a finite nonempty set. As usual we denote by $A^*$ the set of all finite words over $A$ and let $A^+ = A^* \setminus \{\varepsilon\}$ where $\varepsilon$ denotes the empty word. Given two words $w, u \in A^*$, the shuffle product $w \shuffle u$ of $w$ and $u$ is a language over $A$ defined by
$$w \shuffle u= \{w_1u_1 \ldots w_mu_m \mid w_1, \ldots, w_m, u_1, \ldots , u_m \in A^* \text{ and } w=w_1 \ldots w_m, u=u_1\ldots u_m \}.$$

\subsection{Semirings}
A \emph{semiring} $(K,+,\cdot,0,1)$\emph{ }consists of a set
$K,$ two binary operations $+$ and $\cdot$ and two constant elements $0$
and $1$ such that $(K,+,0)$ is a commutative monoid, $(K,\cdot,1)$ is a
monoid, multiplication $\cdot$ distributes over addition $+$, and $0 \cdot k=k\cdot 0=0$
for every $k\in K$. If the monoid $(K,\cdot,1)$ is commutative, then the
semiring is called commutative.\ The semiring is denoted simply by $K$ if the
operations and the constant elements are understood. The result of the empty
product as usual equals to $1$. If no confusion arises, we denote sometimes in the sequel the multiplication operation $\cdot$ just by juxtaposition.\hide{The semiring $K$ is called (additively) idempotent if
$k + k=k$ for every $k\in K$.} The following algebraic structures are well-known commutative semirings. \begin{itemize}
\item The semiring $(\mathbb{N},+,\cdot,0,1)$  of natural numbers,
\item the semiring  $(\mathbb{Q},+,\cdot,0,1)$ of rational numbers,
\item the semiring  $(\mathbb{R},+,\cdot,0,1)$ of real numbers,
\item the Boolean semiring $B=(\{0,1\},+,\cdot,0,1)$, 
\item the arctical or max-plus semiring $\mathbb{R}_{\max}=(\mathbb{R}_{+}\cup\{-\infty\},\max,+,-\infty,0)$ where $\mathbb{R}_{+}=\{r\in\mathbb{R}\mid r\geq0\}$,
\item the tropical or min-plus semiring $\mathbb{R}_{\min}=(\mathbb{R}_{+}\cup\{\infty\},\min,+,\infty,0)$,
\item the Viterbi semiring $(\left[  0,1\right]  ,\max,\cdot,0,1)$ used in probability theory, 
\item every bounded distributive lattice with the operations $\sup$ and $\inf$, in particular the fuzzy semiring $F=([0,1],\max,\min,0,1)$.
\end{itemize}

\hide{\noindent All the aforementioned semirings are commutative, and all but the
first two are idempotent.}
A semiring $(K,+,\cdot , 0,1)$ is called a \emph{skew field} if the $(K,+,0)$ and $(K\setminus \{0\}, \cdot, 1)$ are groups. For instance, $(\mathbb{Q}, +, \cdot,0, 1)$ and  $(\mathbb{R}, +, \cdot, 0, 1)$ are skew fields, and more generally every field is a skew field.

 A  \emph{formal series} (or simply \emph{series}) \emph{over}
$A^*$ \emph{and} $K$ is a mapping $s:A^*\rightarrow K$. The \emph{support of} $s$ is the set $\mathrm{supp}(s)=\{w \in A^* \mid s(w) \neq 0  \}$. A series with finite support is called  a \emph{polynomial}. The \emph{constant series}
$\widetilde{k}$ ($k\in K$) is defined, for every $w\in A^*$,$\ $by
$\widetilde{k}(w)=k$. We denote by $K\left\langle \left\langle A^*\right\rangle
\right\rangle $ the class of all series over $A^*$ and $K$, and by $K\left\langle  A^*
\right\rangle $ the class of all polynomials over $A^*$ and $K$. 
Let $s,r\in K\left\langle \left\langle A^*\right\rangle \right\rangle $ and
$k\in K$. The \emph{sum} $s\oplus r$, the \emph{products with scalars} $ks$
and $sk$, and the\emph{ Hadamard product} $s \otimes r$\ are series in $K\left\langle \left\langle A^*\right\rangle \right\rangle $ and defined
elementwise, respectively  by $ s \oplus r(w)=s(w) + r(w), (ks)(w)=k \cdot s(w), (sk)(w)=s(w)\cdot k, s\otimes r(w)=s(w)\cdot r(w)$  for every
$w\in A^*$. It is a folklore result that the structure $\left(  K\left\langle
\left\langle A^*\right\rangle \right\rangle ,\oplus,\otimes,\widetilde
{0},\widetilde{1}\right)$  is a semiring. Moreover, if $K$ is commutative,
then $\left(  K\left\langle \left\langle A^*\right\rangle \right\rangle
\oplus,\otimes,\widetilde{0},\widetilde{1}\right)  $ is also commutative. The \emph{Cauchy product }$s\odot
r\in K\left\langle \left\langle A^{\ast}\right\rangle \right\rangle $ is determined by $s\odot r(w)=\sum\nolimits_{w=w_1w_2}s(w_1)r(w_2)$ for every $w\in A^{\ast}$,  whereas the \emph{shuffle product} $s\varpi r \in K\left\langle \left\langle A^{\ast}\right\rangle \right\rangle $ is defined by $s\varpi r(w)=\sum\nolimits_{w\in w_1\shuffle w_2}s(w_1)r(w_2)$ for every $w\in A^{\ast}$. \hide{Finally, let $L\subseteq A^*$. Then, the \emph{characteristic series} $1_L \in K\left\langle \left\langle A^{\ast}\right\rangle \right\rangle $ \emph{of} $L$ is determined by $(1_L,w)=1$ if $w \in L$, and $(1_L,w)=0$ otherwise for every $w \in A^*$.}

\begin{quote}
\emph{Throughout the paper }$(K, +, \cdot, 0,1)$\emph{ will denote a commutative semiring.}
\end{quote}

\subsection{Extended propositional interaction logic}
In this subsection we recall from \cite{Pi:Ar} the extended propositional interaction logic. With that logic, we succeeded to describe the order of execution of interactions required by specific architectures. We need to recall firstly propositional interaction logic (PIL for short) (cf. \cite{Bo:Ch,Bo:St,Ma:Co}).  
 
Let $P$ be a finite nonempty set of \emph{ports}. We let $I(P)=\mathcal{P}%
(P)\setminus\{\emptyset\}$ for the set of interactions over $P$.\hide{and  $\Gamma(P)=\mathcal{P}(I(P))\setminus \{\emptyset\}$.Every set $a\in I(P)$ is called an
\emph{interaction} and every set $\gamma\in \Gamma(P)$ is called a set of interactions.}
Then, the syntax of PIL formulas $\phi$ over $P$\ is given by the grammar
$$
\phi::=\mathrm{true}\mid p\mid\neg\phi\mid\phi\vee\phi 
$$
where $p\in P$. 

We set $\neg(\neg\phi)=\phi$ for every
PIL formula $\phi$ and $\mathrm{false}=\neg\mathrm{true}$. As usual the conjunction and the implication 
of two PIL formulas $\phi,\phi^{\prime}$ over $P$ are defined respectively, by $\phi\wedge
\phi^{\prime}:=\neg{\left(  \neg\phi\vee\neg\phi^{\prime}\right)  }$ and $\phi \rightarrow \phi': = \neg \phi \vee \phi'$.  
PIL formulas are interpreted over interactions in $I(P)$. More precisely, for every PIL formula  $\phi$ and $a \in I(P)$ we define the satisfaction relation 
$a\models_{\mathrm{PIL}}\phi$ by induction on the structure of $\phi$ as follows:
\begin{itemize}
\item[-] $a\models_{\mathrm{PIL}} \mathrm{true}$,

\item[-] $a\models_{\mathrm{PIL}} p$  \ \ iff \ \  $p \in a$,

\item[-] $a\models_{\mathrm{PIL}} \neg\phi$ \ \ iff \ \ $a \not \models_{\mathrm{PIL}} \phi$,

\item[-] $a\models_{\mathrm{PIL}} \phi_1 \vee \phi_2$ \ \ iff \ \ $a \models_{\mathrm{PIL}} \phi_1$ or $a \models_{\mathrm{PIL}} \phi_2$.
\end{itemize}

\noindent Two PIL formulas $\phi, \phi'$ are called equivalent, and we denote it by $\phi \equiv \phi'$, whenever $a \models \phi$ iff $a \models \phi'$ for every $a \in I(P)$. A PIL formula $\phi$ is called a \emph{monomial over} $P$ if it is of the form $p_1\wedge \ldots \wedge  p_l$, where $l\geq 1$ and $p_{\lambda}\in P$ or $p_{\lambda}=\neg p'_{\lambda}$ with $p'_{\lambda}\in P$, for every $\lambda\in [l]$. 
For every interaction $a=\lbrace p_1,\ldots,p_l\rbrace \in I(P)$ we consider the monomial $\phi_a =p_1\wedge \ldots \wedge  p_l$. Then, it trivially holds $a\models_{\mathrm{PIL}} \phi_a$, and  for every $a, a' \in I(P)$ we get $a=a'$ iff $\phi_{a} \equiv \phi_{a'} $. \hide{We can describe a set of interactions as a disjunction of PIL formulas. More precisely, let $\gamma=\lbrace a_1,\ldots, a_m\rbrace$ be a set of interactions, where $a_{\mu}=\left \{ p_1^{(\mu)},\ldots,p^{(\mu)}_{l_{\mu}}\right \} \in I(P)$ for every $\mu\in [m]$. Then, the PIL formula $\phi_{\gamma} $ of $\gamma$ is $\phi_{\gamma}=\phi_{a_1}\vee \ldots \vee \phi_{a_m}$, i.e., 
$\phi_{\gamma}=\bigvee\limits_{\mu\in [m]}\bigwedge\limits_{\lambda\in [l_{\mu}]}p_{\lambda}^{(\mu)}$. 

We say that a PIL formula $\phi$ is in \emph{disjunctive normal form} (DNF for short) if 
$$\phi= \bigvee\limits_{\mu\in [m]}\bigwedge\limits_{\lambda\in [l_{\mu}]}p_{\lambda}^{(\mu)}$$
 where $p_1^{(\mu)}\wedge \ldots \wedge p^{(\mu)}_{l_{\mu}}$ is a monomial over $P$, for every $\mu\in [m]$. \\
It is well known that for every Boolean formula (and hence PIL formula $\phi$ over $P$) we can effectively construct an equivalent one in DNF \cite{Da:In}.}

Next we recall the \emph{extended propositional interaction logic} (EPIL for short) (cf. \cite{Pi:Ar}).
Let $P$ be a finite set of ports. The syntax of EPIL
formulas $\varphi$ over $P$\ is given by the grammar
\begin{align*}
\zeta & ::=\phi \mid \zeta * \zeta \\
\varphi & ::=\zeta \mid \neg \zeta \mid \varphi\vee\varphi \mid \varphi \wedge \varphi \mid \varphi * \varphi
\mid \varphi \shuffle \varphi
\end{align*}
where $\phi $ is a PIL formula over $P$, $*$ is the concatenation operator, and $\shuffle$ is the shuffle operator.

The binding strength, in decreasing order, of the operators in EPIL is the following: negation, shuffle, concatenation, conjunction, and disjunction.  The reader should notice that we consider a restricted  use of  negation in the syntax of EPIL formulas. Specifically, negation is permitted in PIL formulas and EPIL formulas of type $\zeta$. The latter ensures exclusion of erroneous interactions in architectures.  The restricted use of negation has no impact to description of models characterized by EPIL formulas since most of  known architectures can be described by sentences in our EPIL. Furthermore, it contributed to a reasonable complexity
of translation of first-order extended interaction logic formulas to finite automata. This in
turn implied the decidability of equivalence, satisfiability, and validity of first-order extended
interaction logic sentences (cf. \cite{Pi:Ar}).
  
For the satisfaction of EPIL formulas we consider finite words $w$ over $I(P)$. Intuitively, a word $w$ encodes each of the distinct interactions within a system as a letter. Moreover, the position of each letter in $w$ depicts the order in which the corresponding interaction is executed in the system, in case there is an order restriction.

\begin{definition}\label{sem-epil}
Let $\varphi$ be an \emph{EPIL} formula over $P$ and $w \in I(P)^*$. If $w=\varepsilon$ and $\varphi=\mathrm{true}$, then we set $w \models \mathrm{true}$. If $w \in I(P)^+$, then we define the satisfaction relation 
$w\models\varphi$ by induction on the structure of $\varphi$ as follows:
\begin{itemize}
\item[-] $w\models \phi$  \ \ iff \ \ $ w \models_{\mathrm{PIL}} \phi$,

\item[-] $w \models \zeta_1 * \zeta_2$  \ \ iff \ \ there exist $w_1,w_2 \in I(P)^*$ such that  $w=w_1w_2$ and $w_i \models\zeta_i$ for $i=1,2$,

\item[-] $w \models \neg \zeta$ \ \ iff \ \ $w \not \models \zeta$,

\item[-] $w\models \varphi_1 \vee \varphi_2$ \ \ iff \ \ $w \models \varphi_1$ or $w \models \varphi_2$,

\item[-] $w\models \varphi_1 \wedge \varphi_2$ \ \ iff \ \ $w \models \varphi_1$ and $w \models \varphi_2$,

\item[-] $w \models \varphi_1 * \varphi_2$  \ \ iff \ \ there exist $w_1,w_2 \in I(P)^*$ such that  $w=w_1w_2$ and $w_i \models\varphi_i$ for $i=1,2$,

\item[-] $w \models \varphi_1 \shuffle \varphi_2$  \ \ iff \ \ there exist $w_1,w_2 \in I(P)^*$ such that  $w\in w_1\shuffle w_2$ and $w_i \models\varphi_i$ for $i=1,2$.

\end{itemize}
\end{definition}

If $\varphi=\phi$ is a PIL formula, then $w \models \phi$ implies that $w$ is a letter in $I(P)$. Two EPIL formulas $\varphi, \varphi'$ are called \emph{equivalent}, and we denote it by $\varphi \equiv \varphi' $, whenever $w \models \varphi$ iff $w \models \varphi'$ for every $w \in I(P)^*$.

\subsection{Component-based systems}
\label{comp_based_system}
We recall the concept of component-based systems which are comprised of a finite number of components of the same or different type. In our set up, components  are defined by labelled transition systems (LTS for short) like in several well-known component-based modelling frameworks including BIP \cite{Bl:Al,Tr:Di}, REO \cite{Am:RE}, X-MAN \cite{He:Co}, and B \cite{Al:Th}. We use the terminology of BIP framework for the basic notions and definitions, though we focus only on the communication patterns of components building a component-based system.  Communication among components is achieved through their corresponding interfaces. For an LTS, its interface is the associated set of labels, called ports. Then, communication of components is defined by interactions, i.e., sets of ports, or equivalently by PIL formulas. In \cite{Pi:Ar} we replaced PIL formulas by EPIL formulas. Hence, we succeeded to describe the execution order of interactions required by the underlying architecture of every component-based system.  

\begin{definition}\label{at_comp}
 An \emph{atomic component} is an \emph{LTS} $B=(Q,P,q_0,R)$ where $Q$ is a finite
set of \emph{states}, $P$ is a finite set of \emph{ports}, $q_0$ is the \emph{initial state} and $R\subseteq Q \times P \times Q$ is the set of \emph{transitions}.
\end{definition}

For simplicity we assume that every port in $P$ occurs in at most one transition in $R$. This simplification has been also considered by other authors (cf. for instance \cite{Bo:Ch,Bo:St}). In the sequel, we call an atomic component $B$ a \emph{component}, whenever we deal with several atomic components. For every set $\mathcal{B} = \{B(i) \mid   i \in [n] \} $ of components, with $B(i)=(Q(i),P(i),q_{0}(i),R(i))$, $i \in [n]$, we consider in the paper, we assume that  $(Q(i) \cup P(i))\cap (Q(i') \cup P(i')) = \emptyset $ 
for every $1 \leq i\neq i' \leq n$. 

Let $\mathcal{B} = \{B(i) \mid i \in [n]   \} $ be a set of components. We let $P_{\mathcal{B}} =\bigcup_{i \in [n]}  P(i)$ comprising all ports of the elements of $\mathcal{B}$. Then an \emph{interaction of} $\B$ is an  interaction $a \in I(P_{\mathcal{B}})$ such that  $\vert a \cap P(i)\vert \leq 1$, for every $i \in [n]$. If $p \in a$, then we say that $p$ is active in $a$. We denote by $I_{\mathcal{B}}$ the set of all interactions of $\mathcal{B}$, i.e., 
$$I_{\mathcal{B}} = \left \{ a \in I(P_{\mathcal{B}}) \mid   \vert a \cap P(i)\vert \leq 1 \text{ for every } i \in [n] \right\}.$$

\begin{definition} \label{BIP-def}A \emph{component-based system} is a pair $(\B, \varphi)$ where $\B=\{B(i) \mid i \in [n]  \}$ is a set of components, with $B(i)=(Q(i),P(i),q_{0}(i),R(i))$ for every $i \in [n]$, and $\varphi$  is an \emph{EPIL} formula over  $P_{\B}$.
\end{definition}

In the above definition the EPIL formula $\varphi$ is defined over the set of ports $P_{\B}$ and  is interpreted over words in $I_{\B}^*$.

\section{Weighted EPIL and component-based systems}
In this section, we introduce the notion of weighted EPIL over a set of ports $P$ and the semiring $K$. Furthermore, we  define weighted component-based systems and provide examples of architectures with quantitative characteristics.

\begin{definition}
Let $P$ be a finite set of ports. Then the syntax of \emph{weighted EPIL} (\emph{wEPIL} for short) \emph{formulas over} $P$ \emph{and} $K$ is given by the grammar
$$\tilde{\varphi}::= k \mid \varphi \mid \tilde{\varphi}_1 \oplus  \tilde{\varphi}_2 \mid \tilde{\varphi}_1 \otimes  \tilde{\varphi}_2 \mid \tilde{\varphi}_1 \odot  \tilde{\varphi}_2 \mid \tilde{\varphi}_1 \varpi  \tilde{\varphi}_2  $$ 
where $k \in K$, $\varphi$ is an \emph{EPIL} formula over $P$, and $\oplus$, $\otimes$, $\odot$, and $\varpi$ are the weighted disjunction, conjunction, concatenation, and shuffle operator, respectively.
\end{definition}

If $\tilde{\varphi}$ is composed by elements in $K$ and PIL formulas connected with $\oplus$ and $\otimes$ operators only, then it is called also a weighted PIL (wPIL for short) formula over $P$ and $K$ \cite{Pa:On,Pa:We} and it will be denoted also by $\tilde{\phi}$. 

For the semantics of wEPIL formulas  we consider finite words $w$ over $I(P)$ and interpret wEPIL formulas as series in $K \left \langle \left \langle I(P)^* \right \rangle \right \rangle $.

\begin{definition}\label{wsem-epil}
Let $\e$ be a \emph{wEPIL} formula over $P$ and $K$. The semantics of $\e$ is a series $\left \Vert \e \right \Vert \in K \left \langle \left \langle I(P)^* \right \rangle \right \rangle$.  For every $w \in I(P)^*$ the value $\left \Vert \e \right \Vert (w)$ is defined inductively on the structure of $\e$ as follows: 
\begin{itemize}
\item[-] $\left \Vert k \right \Vert(w) = k$,

\item[-] $\left \Vert \varphi \right \Vert(w) = \left\{
\begin{array}
[c]{rl}%
1 & \textnormal{ if }w\models \varphi\\
0 & \textnormal{ otherwise}%
\end{array}
,\right.  $

\item[-] $\left \Vert \e_1 \oplus \e_2 \right \Vert(w)=\left \Vert \e_1  \right \Vert(w)+ \left \Vert  \e_2 \right \Vert(w)$,

\item[-] $\left \Vert \e_1 \otimes \e_2 \right \Vert(w)=\left \Vert \e_1  \right \Vert(w) \cdot \left \Vert  \e_2 \right \Vert(w)$,

\item[-] $\left \Vert \e_1 \odot \e_2 \right \Vert(w)=\sum\limits_{w=w_1w_2} ( \left \Vert \e_1  \right \Vert(w_1) \cdot \left \Vert  \e_2 \right \Vert(w_2))$,

\item[-] $\left \Vert \e_1 \varpi \e_2 \right \Vert(w)=\sum\limits_{w \in w_1 \shuffle w_2} ( \left \Vert \e_1  \right \Vert(w_1) \cdot \left \Vert  \e_2 \right \Vert(w_2))$.
\end{itemize}
\end{definition}

Next we define weighted component-based systems. For this, we introduce the notion of a weighted atomic component.

\begin{definition} A \emph{weighted atomic component over} $K$ is a pair $wB=(B, wt)$ where  $B=(Q,P,q_0,R)$ is an atomic component and $wt:R\rightarrow K$ is a weight mapping.
\end{definition}

Since every port in $P$ occurs in at most one transition in $R$, we consider, in the sequel, $wt$ as a mapping $wt:P \rightarrow K$.
If a port $p\in P$ occurs in no transition, then we set $wt(p)=0$.

We call a weighted atomic component $wB$ over $K$ a \emph{weighted component}, whenever we deal with several weighted atomic components and the semiring $K$ is understood. Let $w\B = \{wB(i) \mid   i \in [n] \} $ be a set of weighted components where $wB(i)=(B(i), wt(i))$ with $B(i)=(Q(i),P(i),q_{0}(i), R(i)) $ for every $i \in [n]$. The set of ports and the set of interactions of $w\B$  are the sets $P_{\mathcal{B}}$ and $I_{\mathcal{B}}$ respectively, of the underlying set of components $\B =\{B(i) \mid i \in [n] \}$. Let $a=\{p_{j_1},\ldots,p_{j_m}\}$ be an interaction in $I_{\B}$ such that $p_{j_l}\in P(j_l)$ for every $l\in [m]$. Then, the weighted monomial $\tilde{\phi}_a$ of $a$ is given by the wPIL formula 
\begin{align*}
\tilde{\phi}_a &=( wt(j_1)(p_{j_1})\otimes p_{j_1}) \otimes \ldots \otimes (wt(j_m)(p_{j_m}) \otimes p_{j_m}) \\
& \equiv (wt(j_1)(p_{j_1})\otimes   \ldots \otimes wt(j_m)(p_{j_m}))  \otimes  (p_{j_1}\otimes\ldots     \otimes p_{j_m}) \\
& \equiv (wt(j_1)(p_{j_1})\otimes  \ldots \otimes wt(j_m)(p_{j_m}))  \otimes  (p_{j_1} \wedge \ldots  \wedge  p_{j_m})
\end{align*}
where the first equivalence holds since $K$ is commutative and the second one since $p \otimes p' \equiv p\wedge p'$ for every $p,p' \in P_{\B}$.

\begin{definition}
A \emph{weighted component-based system (over} $K$\emph{)} is a pair $(w\B, \tilde{\varphi})$ where $w\B =\{wB(i) \mid i\in [n] \}$ is a set of weighted components and $\e$ is a \emph{wEPIL} formula over $P_{\B}$ and $K$. 
\end{definition}

We should note that, as in the unweighted case,  the wEPIL formula $\e$ is defined over the set of ports $P_{\B}$ and  is interpreted as a series in $K \left \langle \left \langle I_{\B}^* \right \rangle \right \rangle $.

\hide{
Let $\gamma= \{a_1, \ldots, a_m \}$ be a set of interactions in $I_{\B}$. Then the wPIL formula $\e_{\gamma}$ is defined by $\e_{\gamma} = \e_{a_1} \oplus \ldots \oplus \e_{a_m}$, where $\e_{a_i}$ is defined for every $i \in [n]$ as follows. Let us assume that $a_i=\{p_{j_1}, \ldots , p_{j_r}\}$ where $j_1, \ldots , j_r$ are pairwise distinct and $p_{j_l} \in P(j_l)$ for every $\l \in [r]$. Then $\e_{a_i}=wt(j_1)(p_{j_1})\odot p_{j_1}\odot \ldots \odot wt(j_r)(p_{j_r})\odot p_{j_r}$. Since the semiring $K$ is commutative $\e_{a_i}$ is equivalent to  $wt(j_1)(p_{j_1}) \odot  \ldots \odot wt(j_r)(p_{j_r}) \odot(p_{j_1} \wedge \ldots \wedge p_{j_r})$, i.e., to $wt(j_1)(p_{j_1}) \odot  \ldots \odot wt(j_r)(p_{j_r}) \odot \varphi_{a_i}$. By definition of the monomials $\varphi_{a_i}$, $i \in [n]$ we get that $\e_{\gamma}(a_i)=  \sum\limits_{a' \in \gamma, a' \subseteq a_i}\e_{a'}(a_i)$. Nevertheless, in the weighted formulas for well-known architectures, as they are presented in Subsection \ref{chap:examp}, the case $a' \in \gamma, a' \subseteq a_i$ does not occur. Hence, in the sequel, we assume that the set $\gamma \in \Gamma_{\B}$ of interactions, in a wEPIL formula $\e_{\gamma}$ assigned to a  set $\wB$ of weighted components, satisfies the statement: if $a, a' \in \gamma$ with $a \neq a'$, then $a\nsubseteq a'$ and $a' \nsubseteq a$. In general if this is not the case, then we can replace in $\e_{\gamma}$ every formula $\e_a$ by its corresponding full monomial. For instance, keeping the notations for $a_i$ we should write $\e_{a_i}=wt(j_1)(p_{j_1})\odot p_{j_1}\odot \ldots \odot wt(j_r)(p_{j_r})\odot p_{j_r} \odot \left(\bigwedge\limits_{p \notin a}\neg p \right) $.}

\subsection{Examples of architectures described by wEPIL formulas}    
In this subsection,  we present three examples of weighted component-based models whose architectures have ordered interactions encoded by wEPIL formulas. We recall from \cite{Pi:Ar} the following macro EPIL formula. Let $P=\{p_1, \ldots , p_n\}$ be a set of ports. Then, for  $p_{i_1}, \ldots , p_{i_m} \in P$ with $m <n$ we let 
$$\#(p_{i_1} \wedge \ldots \wedge p_{i_m})::=p_{i_1}\wedge \ldots \wedge p_{i_m} \wedge \bigwedge_{p \in P \setminus \{p_{i_1}, \ldots, p_{i_m}\}}\neg p.$$ 
Let us now assume that we assign a  weight $k_{i_l} \in K$ to $p_{i_l}$ for every $l \in [m]$. We define the subsequent macro wEPIL formula 
 $$\#_w(p_{i_1} \otimes \ldots \otimes p_{i_m})::=(k_{i_1} \otimes p_{i_1}) \otimes \ldots \otimes (k_{i_m} \otimes p_{i_m}) \otimes \bigwedge_{p \in P \setminus \{p_{i_1}, \ldots, p_{i_m}\}}\neg p. $$
Then, we get \begin{align*}
\#_w(p_{i_1} \otimes \ldots \otimes p_{i_m})& \equiv (k_{i_1} \otimes \ldots \otimes k_{i_m})   \otimes (p_{i_1} \otimes \ldots \otimes p_{i_m}) \otimes \bigwedge_{p \in P \setminus \{p_{i_1}, \ldots, p_{i_m}\}}\neg p \\
& \equiv (k_{i_1} \otimes \ldots \otimes k_{i_m})   \otimes \bigg((p_{i_1} \wedge \ldots \wedge p_{i_m}) \wedge \bigwedge_{p \in P \setminus \{p_{i_1}, \ldots, p_{i_m}\}}\neg p \bigg)\\
& = (k_{i_1} \otimes \ldots \otimes k_{i_m})   \otimes \#(p_{i_1} \wedge \ldots \wedge p_{i_m}).
\end{align*}

\noindent Clearly the above macro formula $\#_w(p_{i_1} \otimes \ldots \otimes p_{i_m})$ depends on the values $k_{i_1}, \ldots, k_{i_m}$. Though, in order to simplify our wEPIL formulas, we make no special notation about this. If the macro formula is defined in a weighted component-based system, then the values $k_{i_1}, \ldots, k_{i_m}$ are unique in the whole formula. 

\begin{example}\label{ex_b_blackboard}\textbf{\emph{(Weighted Blackboard)}}
Blackboard architecture is applied to multi-agent systems 
for solving problems with nondeterministic strategies that result from multiple partial solutions (\hide{cf. Chapter 2 in \cite{Bu:Pa},}\cite{Co:Bl}).
Several applications are based on a blackboard architecture, 
including 
planning and scheduling (cf. \cite{La:Tw,St:Cr}), 
artificial intelligence (cf. \cite{Co:Bl}) and web applications \cite{Li:Vi,Me:We}. Blackboard architecture involves three component types, one blackboard component, one controller
component and the knowledge sources components \cite{Bu:Pa,Co:Bl}. Blackboard is a global data store that presents the state of the problem to be solved. 
Knowledge sources, simply called sources, are expertised agents that provide 
partial solutions to the given problem. Knowledge sources are independent and do not know about the existence of other sources.\hide{In order to offer their solutions, the sources should ensure
that the necessary information (if any) exists on the blackboard.
For this, sources that lie in the library are informed about the current data stored in the blackboard.
Though, the knowledge sources are interested only in specific
information found on blackboard in order to provide their solutions. In order for the sources not to scan all the
existing contributions, the blackboard is responsible for informing the sources whenever the information of
their interest is available.}Whenever
there is sufficient information for a source to provide its partial solution, the corresponding
source is triggered i.e., is keen to write on the blackboard.\hide{ Only triggered sources can write new data on the blackboard.}  Since multiple sources are triggered and compete to provide their solutions, a controller component is used
to resolve any conflicts. Controller accesses both the blackboard to inspect the available data and the sources to schedule them so that they execute their solutions on the blackboard.  

Consider a weighted component-based system  $(w\B, \tilde{\varphi})$ with the Blackboard architecture and three sources. 
Therefore, we have  $w\B=\lbrace wB(1), wC(1),wS(1),wS(2), wS(3)\rbrace$ referring to blackboard,  controller, and the three sources
weighted components, respectively. \hide{The set of ports of each weighted component is $P(1)=\lbrace p_d,p_a\rbrace,
 P(2)=\lbrace p_r,p_l,p_e\rbrace,$ $P(3)=\lbrace p_{n_1},p_{t_1},p_{w_1} \rbrace$, $P(4)=\lbrace p_{n_2},p_{t_2},p_{w_2} \rbrace$, and $P(5)=\lbrace p_{n_3},p_{t_3},p_{w_3} \rbrace$.}Figure \ref{wb_blackboard} depicts a possible execution of the permissible interactions in the system. The weight associated with each port in the system is shown at the outside of the port. Blackboard component has two ports $p_d,p_a$ to declare the state of the problem and add the new data as obtained by a source, respectively.
Sources  have three ports $p_{n_i},p_{t_i},p_{w_i}$, for $i=1,2,3$, for being notified about the existing data
on the blackboard, the trigger of the source, and for writing the partial solution on the blackboard, respectively. 
Controller has three ports, namely $p_r$ used to record blackboard data, $p_l$ for the log process of
 triggered sources, and $p_e$ for their execution to  blackboard. Here we assume that all knowledge sources are triggered, i.e., that all available sources participate in the architecture. The interactions range over $I_{\B}$ and the \emph{wEPIL} formula $\e$ for the weighted Blackboard architecture is 
\begin{multline*}\e=\#_w(p_d\otimes p_r) \odot \bigg(\#_w(p_d\otimes p_{n_1})\varpi \#_w(p_d\otimes p_{n_2})\varpi \#_w(p_d\otimes p_{n_3}) \bigg) \odot \\ \bigg(\e_1\oplus \e_2\oplus \e_3\oplus(\e_1 \varpi \e_2) \oplus (\e_1\varpi \e_3)\oplus (\e_2\varpi \e_3)\oplus (\e_1\varpi\e_2\varpi\e_3)\bigg)
\end{multline*}
where 
$$\e_i= \#_w(p_l\otimes p_{t_i}) \odot \#_w(p_e\otimes p_{w_i}\otimes p_a)$$
for $i=1,2,3$. The first \emph{wPIL} subformula expresses the cost for the connection of the blackboard and controller. The \emph{wEPIL} subformula between the two weighted concatenation operators represents the cost of the connection of the three  sources to blackboard in order to be informed for existing data. The last part of $\e$ captures the cost of applying the connection of some of the three sources with controller and blackboard for the triggering and writing process. 

\noindent Let $w_1=\lbrace p_d, p_r\rbrace\lbrace p_d, p_{n_1}\rbrace\lbrace p_d, p_{n_2}\rbrace\lbrace p_d, p_{n_3}\rbrace\lbrace p_l, p_{t_2}\rbrace\lbrace p_l, p_{t_3}\rbrace\lbrace p_e, p_{w_2},p_a\rbrace\lbrace p_e, p_{w_3},p_a\rbrace$ and $w_2=\lbrace p_d, p_r\rbrace\lbrace p_d, p_{n_3}\rbrace\lbrace p_d, p_{n_1}\rbrace\lbrace p_d, p_{n_2}\rbrace\lbrace p_l, p_{t_3}\rbrace\lbrace p_e, p_{w_3},p_a\rbrace$ so that in $w_1$ the source $wS(2)$ is triggered and writes data before the third source $wS(3)$, and in $w_2$ only source $wS(3)$ is triggered and writes data. The values $\left \Vert \e \right \Vert(w_1)$ and $\left \Vert \e \right \Vert(w_2)$ represent the cost for executing the interactions with the order encoded by $w_1$ and $w_2$, respectively. Then, $\left \Vert \e \right \Vert(w_1)+\left \Vert \e \right \Vert(w_2)$  is the `total' cost for implementing $w_1$ and $w_2$. For instance, in 
the min-plus semiring the value $\min \lbrace \left \Vert \e \right \Vert(w_1),\left \Vert \e \right \Vert(w_2)\rbrace$ gives information for
the communication with the minimum cost. On the other hand, in Viterbi semiring, we get
the value $\max \lbrace \left \Vert \e \right \Vert(w_1),\left \Vert \e \right \Vert(w_2)\rbrace$ which refers to the communication with the maximum
probability to be executed.

\definecolor{harlequin}{rgb}{0.25, 1.0, 0.0}
\definecolor{ao}{rgb}{0.0, 0.0, 1.0}

\definecolor{darkorange}{rgb}{1.0, 0.55, 0.0}
\definecolor{harlequin}{rgb}{0.25, 1.0, 0.0}
\definecolor{ao}{rgb}{0.0, 0.0, 1.0}

\begin{figure}[h]
\centering
\resizebox{0.65\linewidth}{!}{
\begin{tikzpicture}[>=stealth',shorten >=1pt,auto,node distance=1cm,baseline=(current bounding box.north)]
\tikzstyle{component}=[rectangle,ultra thin,draw=black!75,align=center,inner sep=9pt,minimum size=1.5cm,minimum height=3.4
cm,minimum width=2.6cm]
\tikzstyle{port}=[rectangle,ultra thin,draw=black!75,minimum size=6mm]
\tikzstyle{bubble} = [fill,shape=circle,minimum size=5pt,inner sep=0pt]
\tikzstyle{type} = [draw=none,fill=none] 

\node [component,align=center] (a1)  {};
\node [port] (a2) [above right= -1.6cm and -0.62cm of a1]  {$p_d$};
\node[bubble] (a3) [above right=-1.35cm and -0.08cm of a1]   {};

\node[]       (k1) [left= 0.4cm of a3]{$k_d$};

\node [port] (a4) [above right= -2.8cm and -0.63cm of a1]  {$p_a$};
\node[bubble] (a5) [above right=-2.57cm and -0.08cm of a1]   {};

\node[]       (k2) [left= 0.4cm of a5]{$k_a$};

\node[type] (a6) [above=-0.45cm of a1]{Blackb.};
\node            [below=-0.1cm of a6]{$wB(1)$};

\tikzstyle{control}=[rectangle,ultra thin,draw=black!75,align=center,inner sep=9pt,minimum size=1.5cm,minimum height=2.3cm,minimum width=3.5cm]

\node [control,align=center] (b1)  [below right=8cm and 2.5cm of a1]{};
\node [port] (b2) [above left=-0.605cm  and -1.05 of b1]  {$p_r$};
\node[bubble] (b3) [above left=-0.1cm and -0.35cm of b2]   {};

\node[]       (k3) [below= 0.4cm of b3]{$k_r$};

\node [port] (b4) [right=0.4 of b2]  {$p_l$};
\node[bubble] (b5) [above left=-0.1cm and -0.35cm of b4]   {};

\node[]       (k4) [below= 0.4cm of b5]{$k_l$};

\node [port] (b6) [right=0.4 of b4]  {$p_e$};
\node[bubble] (b7) [above left=-0.1cm and -0.35cm of b6]   {};

\node[]       (k5) [below= 0.4cm of b7]{$k_e$};

\node[type] (b8) [below=-1.2cm of b1]{Contr.};
\node[]      [below=-0.1cm of b8]{$wC(1)$};

\tikzstyle{source}=[rectangle,ultra thin,draw=black!75,align=center,inner sep=9pt,minimum size=1.5cm,minimum height=3.5cm,minimum width=3cm]

\tikzstyle{port}=[rectangle,ultra thin,draw=black!75,minimum size=6mm,inner sep=3pt]

\node [source,align=center] (c1) [above right=-1cm and 9cm of a1] {};
\node [port] (c2) [above left= -0.95cm and -0.75cm of c1]  {$p_{n_1}$};
\node[bubble] (c3) [above left=-0.35cm and -0.08cm of c2]   {};

\node[] (k6) [right= 0.6cm of c3]{$k_{n_1}$};

\tikzstyle{port}=[rectangle,ultra thin,draw=black!75,minimum size=6mm,inner sep=3.5pt]
\node [port] (c4) [above left= -1.6cm and -0.72cm of c2]  {$p_{t_1}$};
\node[bubble] (c5) [above left=-0.35cm and -0.08cm of c4]   {};

\node[] (k7) [right= 0.6cm of c5]{$k_{t_1}$};

\tikzstyle{port}=[rectangle,ultra thin,draw=black!75,minimum size=6mm,inner sep=2.7pt]
\node [port] (c6) [above left= -1.6cm and -0.72cm of c4]  {{\small $p_{w_1}$}};
\node[bubble] (c7) [above left=-0.35cm and -0.08cm of c6]   {};

\node[] (k8) [right= 0.6cm of c7]{$k_{w_1}$};

\node[type,align=center] (c8) [above right=-0.45cm and -1.6cm of c1]{Sour. $1$};
\node[] [below=-0.1cm of c8]{$wS(1)$};

\tikzstyle{port}=[rectangle,ultra thin,draw=black!75,minimum size=6mm,inner sep=3pt]
\node [source,align=center] (d1) [below= 1cm of c1] {};
\node [port] (d2) [above left= -0.95cm and -0.75cm of d1]  {$p_{n_2}$};
\node[bubble] (d3) [above left=-0.35cm and -0.08cm of d2]   {};

\node[] (k9) [right= 0.6cm of d3]{$k_{n_2}$};

\tikzstyle{port}=[rectangle,ultra thin,draw=black!75,minimum size=6mm,inner sep=3.5pt]
\node [port] (d4) [above left= -1.6cm and -0.72cm of d2]  {$p_{t_2}$};
\node[bubble] (d5) [above left=-0.35cm and -0.08cm of d4]   {};

\node[] (k10) [right= 0.6cm of d5]{$k_{t_2}$};

\tikzstyle{port}=[rectangle,ultra thin,draw=black!75,minimum size=6mm,inner sep=2.7pt]
\node [port] (d6) [above left= -1.6cm and -0.72cm of d4]  {{\small $p_{w_2}$}};
\node[bubble] (d7) [above left=-0.35cm and -0.08cm of d6]   {};

\node[] (k11) [right= 0.6cm of d7]{$k_{w_2}$};

\node[type] (d8) [above right=-0.45cm and -1.6cm of d1]{Sour. $2$};
\node[] [below=-0.1cm of d8]{$wS(2)$};

\tikzstyle{port}=[rectangle,ultra thin,draw=black!75,minimum size=6mm,inner sep=3pt]
\node [source,align=center] (e1) [below= 1cm of d1] {};
\node [port] (e2) [above left= -0.95cm and -0.75cm of e1]  {$p_{n_3}$};
\node[bubble] (e3) [above left=-0.35cm and -0.08cm of e2]   {};

\node[] (k12) [right= 0.6cm of e3]{$k_{n_3}$};

\tikzstyle{port}=[rectangle,ultra thin,draw=black!75,minimum size=6mm,inner sep=3.5pt]
\node [port] (e4) [above left= -1.6cm and -0.72cm of e2]  {$p_{t_3}$};
\node[bubble] (e5) [above left=-0.35cm and -0.08cm of e4]   {};

\node[] (k13) [right= 0.6cm of e5]{$k_{t_3}$};

\tikzstyle{port}=[rectangle,ultra thin,draw=black!75,minimum size=6mm,inner sep=2.7pt]
\node [port] (e6) [above left= -1.6cm and -0.72cm of e4]  {{\small $p_{w_3}$}};
\node[bubble] (e7) [above left=-0.35cm and -0.08cm of e6]   {};

\node[] (k14) [right= 0.6cm of e7]{$k_{w_3}$};

\node[type] (e8) [above right=-0.45cm and -1.6cm of e1]{Sour. $3$};
\node[] [below=-0.1cm of e8]{$wS(3)$};
 
 
\path[-]          (a3)  edge                  node {} (c3);
 \path[-]          (a3)  edge                  node {} (d3);
 \path[-]          (a3)  edge                  node {} (e3);

 \path[-]          (a3)  edge                  node {} (b3);
  \path[-]          (b3)  edge                  node {} (a3);
 
 
 \path[-]          (c3)  edge                  node {} (a3);
 \path[-]          (d3)  edge                  node {} (a3);
 \path[-]          (e3)  edge                  node {} (a3);

 
 \path[-]          (a5)  edge  [harlequin]                node {} (d7);
 \path[-]          (a5)  edge  [harlequin]                node {} (e7);

 \path[-]          (a5)  edge  [harlequin]                node {} (b7);
 
 
 \path[-]          (d7)  edge   [harlequin]               node {} (a5);
 \path[-]          (e7)  edge   [harlequin]               node {} (a5);

 \path[-]          (b7)  edge   [harlequin]              node {} (a5);
 
  \path[-]          (b5)  edge  [darkorange]                node {} (d5);
   \path[-]          (b5)  edge [darkorange]                node {} (e5);
   
  \path[-]          (d5)  edge  [darkorange]               node {} (b5);
  \path[-]          (e5)  edge  [darkorange]               node {} (b5);

 
  \path[-]          (b7)  edge  [harlequin]               node {} (d7);
 \path[-]          (d7)  edge  [harlequin]               node {} (b7);
 
    \path[-]          (b7)  edge  [harlequin]               node {} (e7);
 \path[-]          (e7)  edge  [harlequin]               node {} (b7);
   
\end{tikzpicture}}
\caption{A possible execution of the interactions in a weighted Blackboard architecture.}
\label{wb_blackboard}
\end{figure}
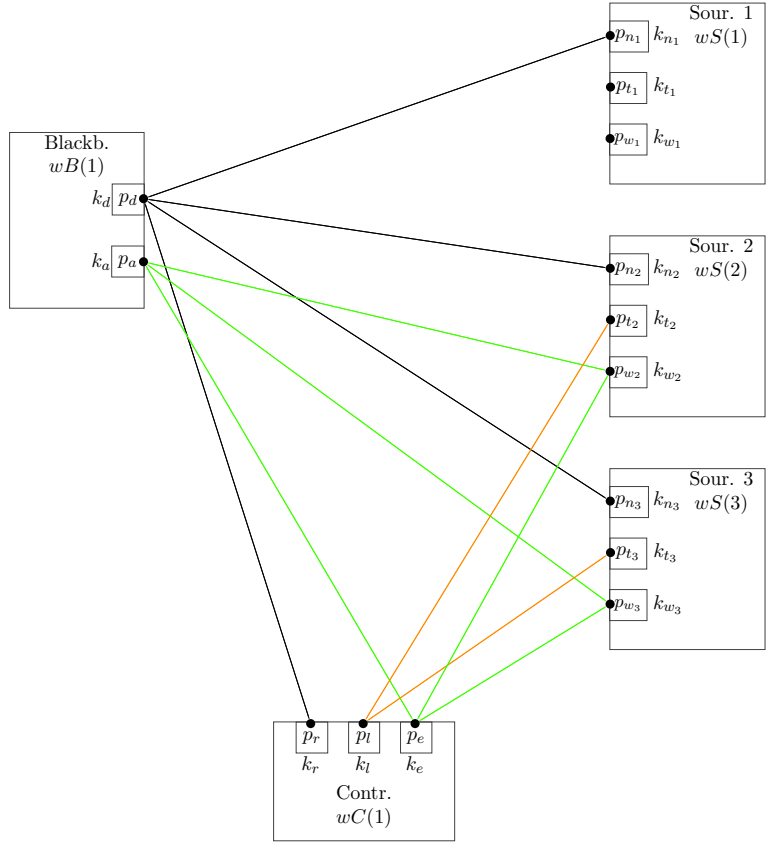
\end{example}

\begin{example} \textbf{\emph{(Weighted Request/Response)}}
\label{b-r-r}
Request/Response architectures represent a classical interaction pattern and are
widely used for web services \cite{Da:Se}. A Request/Response architecture refers to clients and services. Services are offered
by service providers through some common (online) platform. In order for a service to be made available the
service provider needs to subscribe it in the service registry. The enrollment of a service in the registry allows service consumers, simply called clients,
to search the existing services. Once a service is signed up, the client scans the corresponding registry and chooses a service. Then, each 
client that is interested in a service sends a request to the service and waits until the service will respond. No other client can be connected
to the service until the response of the service to the client who sent the request will be completed. In \cite{Ma:Co} the authors described this process
by adding, for each service, a third component called coordinator. Coordinator takes care that only one
client is connected to a service until the process among them is completed. 

\noindent Let $(w\B,\widetilde{\varphi})$ be a weighted component-based system with the Request/Response architecture. We consider four weighted component types namely weighted service registry, service, 
client, and  coordinator. Our weighted system consists of seven components, and specifically, the service registry, two services  with their associated coordinators, and two clients.  
Therefore, we have that $w\B=\lbrace wR(1),wS(1),wS(2),wD(1),wD(2), wC(1),wC(2)\rbrace$ referring to each of the aforementioned weighted components,
respectively. Figure \ref{b_r_r} depicts a case for the permissible interactions in the system. Service registry has three ports denoted by $p_e$, $p_u$, and $p_t$ used for 
 connecting with the service for its enrollment, for authorizing the client to search for a service, and for transmitting the address (link) of the 
service to the client in order for the client to send then its request, respectively. Services have three ports $p_{r_z},p_{g_z}, p_{s_z}$, for $z=1,2$, which establish the connection to the service registry for signing up the service, and the connection to 
a client (via coordinator) for getting a request and responding (sending the response), respectively.
Each client $z$ has five ports denoted by $p_{l_z},p_{o_z},p_{n_z}, p_{q_z}$ and $p_{c_z}$ for $z=1,2$. The first two ports are used for connection with the 
service registry to look up the available services and for obtaining the address of the service that interests the client. 
The latter three ports express the connection of the client to  coordinator, 
 to service (via coordinator) for sending the request, and to service (via coordinator) for collecting its response,
 respectively. Coordinators have three ports namely $p_{m_z}, p_{a_z},p_{d_z}$ for $z=1,2$. The first port controls that only
 one client is connected to a service. The second one is used for acknowledging that the connected client sends a request, and the
 third one disconnects the client when the service responds to the request. The interactions in the architecture range over $I_{\B}$ and the weight of each port is denoted as in Figure \ref{b_r_r}. The \emph{wEPIL} formula $\e$ describing the weighted Request/Response architecture is
\begin{multline*}
\e=\big(\#_w(p_e\otimes p_{r_1})\varpi \#_w(p_e\otimes p_{r_2})\big) \odot \big(\x_1 \varpi \x_2\big) \odot \\ \Bigg( \bigg(\e_{11}  \oplus \e_{21} \oplus (\e_{11} \odot \e_{21}) \oplus (\e_{21} \odot \e_{11})\bigg) \bigoplus  \bigg(\e_{12}  \oplus \e_{22} \oplus (\e_{12} \odot \e_{22}) \oplus (\e_{22} \odot \e_{12}) \bigg) \bigoplus \\ \bigg(\big(\e_{11}  \oplus \e_{21} \oplus (\e_{11} \odot \e_{21}) \oplus (\e_{21} \odot \e_{11})\big) \varpi 
 \\  \big(\e_{12}  \oplus \e_{22} \oplus (\e_{12} \odot \e_{22}) \oplus (\e_{22} \odot \e_{12}) \big)\bigg)\Bigg)
\end{multline*} 
where
\begin{itemize}
 \item[-]$\x_1=\#_w(p_{l_1}\otimes p_u) \odot \#_w(p_{o_1}\otimes p_t)$,
 \item[-]$\x_2=\#_w(p_{l_2}\otimes p_u) \odot \#_w(p_{o_2}\otimes p_t)$,
\end{itemize}
and
\begin{itemize} 
\item[-]$\e_{11}=\#_w(p_{n_1}\otimes p_{m_1}) \odot \#_w(p_{q_1}\otimes p_{a_1}\otimes p_{g_1})\odot \#_w(p_{c_1}\otimes p_{d_1}\otimes p_{s_1})$, 
\item[-] $\e_{12}=\#_w(p_{n_1}\otimes p_{m_2}) \odot \#_w(p_{q_1}\otimes p_{a_2}\otimes p_{g_2})\odot \#_w(p_{c_1}\otimes p_{d_2}\otimes p_{s_2})$, 
\item[-]$\e_{21}=\#_w(p_{n_2}\otimes p_{m_1}) \odot \#_w(p_{q_2}\otimes p_{a_1}\otimes p_{g_1})\odot \#_w(p_{c_2}\otimes p_{d_1}\otimes p_{s_1})$,
\item[-] $\e_{22}=\#_w(p_{n_2}\otimes p_{m_2}) \odot \#_w(p_{q_2}\otimes p_{a_2}\otimes p_{g_2}) \odot \#_w(p_{c_2}\otimes p_{d_2}\otimes p_{s_2})$.
\end{itemize}

\noindent The two \emph{wEPIL} subformulas at the left of the first two weighted concatenation operators encode the cost for the connections of the two services and the two clients with registry, respectively. Then, each of the three \emph{wEPIL} subformulas connected with $\bigoplus$ present the cost for the connection of either one of the two clients or both of them (one at each time) with the first service only, the second service only, or both of the services, respectively.

\noindent Let $w_1=\lbrace p_e,p_{r_1}\rbrace \lbrace p_e,p_{r_2}\rbrace \lbrace p_{l_1},p_u\rbrace  \lbrace p_{l_2},p_u\rbrace\lbrace p_{o_1},p_u\rbrace \lbrace p_{o_2},p_u\rbrace \lbrace p_{n_1}, p_{m_2}\rbrace \lbrace p_{q_1}, p_{a_2}, p_{g_2}\rbrace \lbrace p_{c_1}, p_{d_2},\\ p_{s_2}\rbrace \lbrace p_{n_2}, p_{m_2}\rbrace \lbrace p_{q_2}, p_{a_2}, p_{g_2}\rbrace \lbrace p_{c_2}, p_{d_2}, p_{s_2}\rbrace$ and $w_2=\lbrace p_e,p_{r_2}\rbrace \lbrace p_e,p_{r_1}\rbrace \lbrace p_{l_1},p_u\rbrace  \lbrace p_{l_2},p_u\rbrace\lbrace p_{o_2},p_u\rbrace\\ \lbrace p_{o_1},p_u\rbrace \lbrace p_{n_2}, p_{m_2}\rbrace \lbrace p_{q_2}, p_{a_2}, p_{g_2}\rbrace \lbrace p_{c_2}, p_{d_2}, p_{s_2}\rbrace$. Then $w_1$ encodes one of the possible executions for the interactions in which firstly client $wC(1)$ and then client $wC(2)$ connects via $wD(2)$ to service $wS(2)$, and $w_2$ shows a possible execution for applying the connection only of client $wC(2)$ via $wD(2)$ to $wS(2)$. Then, in max-plus semiring for instance, the value $\max \lbrace \left \Vert \e \right \Vert(w_1),\left \Vert \e \right \Vert(w_2)\rbrace$ gives information for
the communication with the maximum cost.

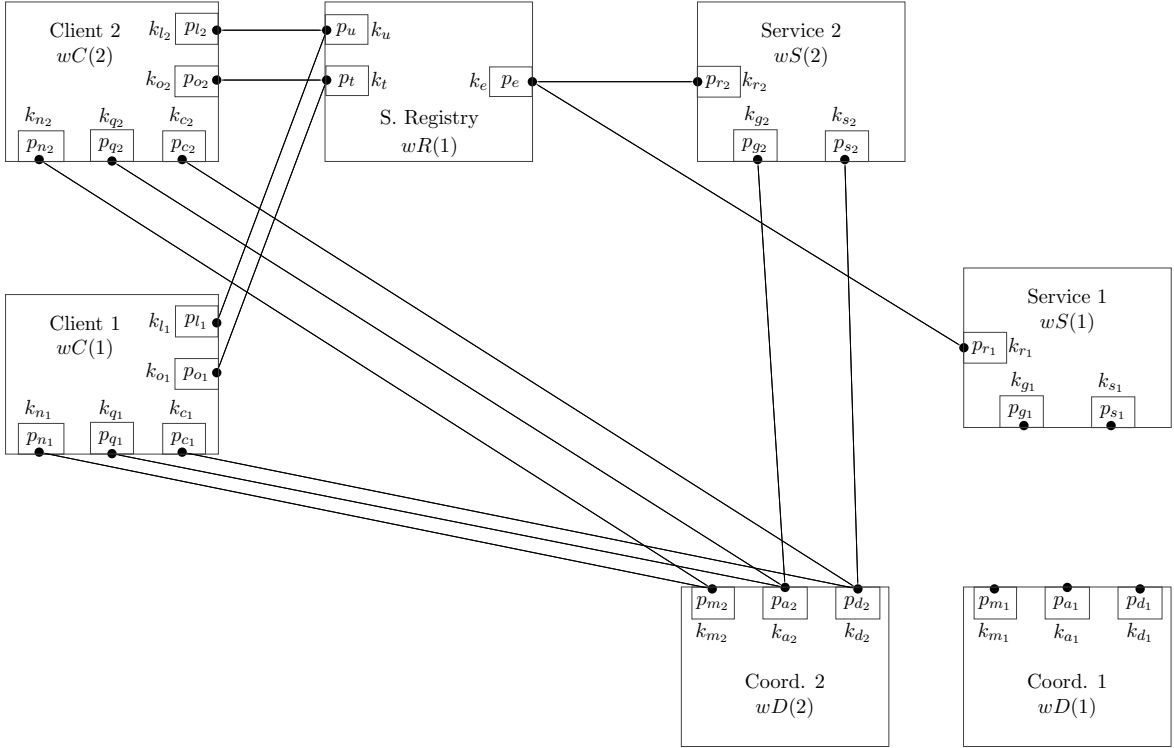
\begin{figure}[h]
\centering
\resizebox{1.0\linewidth}{!}{
\begin{tikzpicture}[>=stealth',shorten >=1pt,auto,node distance=3cm,baseline=(current bounding box.north)]
\tikzstyle{component}=[rectangle,ultra thin,draw=black!75,align=center,inner sep=9pt,minimum size=2.5cm,minimum height=3cm, minimum width=3.9cm]
\tikzstyle{port}=[rectangle,ultra thin,draw=black!75,minimum size=6mm]
\tikzstyle{bubble} = [fill,shape=circle,minimum size=5pt,inner sep=0pt]
\tikzstyle{type} = [draw=none,fill=none]

\node [component] (a1) {};

\tikzstyle{port}=[rectangle,ultra thin,draw=black!75,minimum size=6mm,inner sep=2.5pt]
\node[bubble] (a2)  [above left=-0.105cm and -0.65cm of a1]   {};   
\node [port]  (a3)  [above left=-0.605cm and -1.0cm of a1]  {$p_{m_1}$};  

\node []  (k1)  [below=0.4cm of a2]  {$k_{m_1}$}; 

\tikzstyle{port}=[rectangle,ultra thin,draw=black!75,minimum size=6mm,inner sep=4.5pt]
\node[bubble] (a4)  [above=-0.105cm of a1]   {};   
\node [port]  (a5)  [above=-0.605cm and -2.0cm of a1]  {$p_{a_1}$};  

\node []  (k2)  [below=0.45cm of a4]  {$k_{a_1}$}; 

\node[bubble] (a6)  [above right=-0.105cm and -0.65cm of a1]   {};   
\node [port]  (a7)  [above right=-0.605cm and -1.0cm of a1]  {$p_{d_1}$};  

\node []  (k3)  [below=0.4cm of a6]  {$k_{d_1}$}; 

\node[type]   (a8)      [below=-1.5cm of a1]{Coord. $1$};
\node                   [below=-0.1cm of a8]{$wD(1)$};

\node [component] (h1) [left=1.4 of a1]{};

\tikzstyle{port}=[rectangle,ultra thin,draw=black!75,minimum size=6mm,inner sep=2.5pt]
\node[bubble] (h2)  [above left=-0.105cm and -0.65cm of h1]   {};   
\node [port]  (h3)  [above left=-0.605cm and -1.0cm of h1]  {$p_{m_2}$};  

\node []  (k4)  [below=0.4cm of h2]  {$k_{m_2}$}; 

\tikzstyle{port}=[rectangle,ultra thin,draw=black!75,minimum size=6mm,inner sep=4.5pt]
\node[bubble] (h4)  [above=-0.105cm of h1]   {};   
\node [port]  (h5)  [above=-0.605cm and -2.0cm of h1]  {$p_{a_2}$};  

\node []  (k4)  [below=0.45cm of h4]  {$k_{a_2}$}; 

\node[bubble] (h6)  [above right=-0.105cm and -0.65cm of h1]   {};   
\node [port]  (h7)  [above right=-0.605cm and -1.0cm of h1]  {$p_{d_2}$};  

\node []  (k5)  [below=0.4cm of h6]  {$k_{d_2}$}; 

\node[type]   (h8)      [below=-1.5cm of h1]{Coord. $2$};
\node                   [below=-0.1cm of h8]{$wD(2)$};

\tikzstyle{client}=[rectangle,ultra thin,draw=black!75,align=center,inner sep=9pt,minimum size=2.5cm,minimum height=3.0cm, minimum width=4cm]

\tikzstyle{port}=[rectangle,ultra thin,draw=black!75,minimum size=6mm,inner sep=4.5pt,minimum height=0.4pt]
\node [client] (b1) [above left=8.0cm and 14.0cm of a1]{};

\node[bubble]     (b2) [below left=-0.105cm and -0.7cm of b1]   {};   
\node [port]      (b3) [below left=-0.58cm and -1.1cm of b1]  {$p_{n_2}$}; 
 
\node []  (k6)  [above=0.4cm of b2]  {$k_{n_2}$};

\tikzstyle{port}=[rectangle,ultra thin,draw=black!75,minimum size=6mm,inner sep=4.5pt, minimum height=2pt, minimum width= 0.4cm]
\node[bubble]     (b4) [below=-0.105cm of b1]   {};  
\node [port]      (b5) [below=-0.60cm and -2.0cm of b1]  {$p_{q_2}$}; 

\node []  (k7)  [above=0.4cm of b4]  {$k_{q_2}$};

\tikzstyle{port}=[rectangle,ultra thin,draw=black!75,minimum size=6mm,inner sep=4.5pt, minimum height=1pt, minimum width= 0.4cm]
\node[bubble]     (b6) [below right=-0.105cm and -0.75cm of b1]   {};   
\node [port]      (b7) [below right=-0.58cm and -1.05cm of b1]  {$p_{c_2}$};  
 
 \node []  (k8)  [above=0.4cm of b6]  {$k_{c_2}$};

\tikzstyle{port}=[rectangle,ultra thin,draw=black!75,minimum size=6mm,inner sep=4.5pt, minimum height=2.4pt, minimum width=0.8cm]
\node[bubble] (b8)  [above right=-0.6cm and -0.1cm of b1]   {};   
\node [port]  (b9)  [above right=-0.8cm and -.81cm of b1]  {$p_{l_2}$};  

\node []  (k9)  [left=0.6cm of b8]  {$k_{l_2}$};

\node[bubble] (b10)  [below right=-1.6cm and -0.1cm of b1]   {};   
\node [port]  (b11)  [below right=-1.8 cm and -.82cm of b1]  {$p_{o_2}$};

\node []  (k10)  [left=0.6cm of b10]  {$k_{o_2}$};

\node[type]    (b12)        [above left=-0.8cm and -2.3cm of b1]{Client $2$};
\node                  [below=-0.1cm of b12]{$wC(2)$};

\node [component] (g1) [right=2.0 of b1]{};

\tikzstyle{port}=[rectangle,ultra thin,draw=black!75,minimum size=6mm,inner sep=4.5pt, minimum height=2pt, minimum width=0.8cm]
\node[bubble] (g2)  [right=-0.1cm of g1]   {};   
\node [port]  (g3)  [right=-.81cm of g1]  {$p_{e}$};  

\node []  (k11)  [left=0.6cm of g2]  {$k_{e}$};

\node[bubble] (g4)  [above left=-0.6 and -0.1cm of g1]   {};   
\node [port]  (g5)  [above left=-0.8cm and -.81cm of g1]  {$p_{u}$};  

\node []  (k12)  [right=0.6cm of g4]  {$k_{u}$};

\node[bubble] (g6)  [below left=-1.6cm and -0.1cm of g1]   {};   
\node [port]  (g7)  [below left=-1.8 cm and -.82cm of g1]  {$p_{t}$};  

\node []  (k13)  [right=0.6cm of g6]  {$k_{t}$};

\node[type]   (g8)      [above right=-2.5cm and -3cm of g1]{S. Registry};
\node                   [below=-0.1cm of g8]{$wR(1)$};

\tikzstyle{port}=[rectangle,ultra thin,draw=black!75,minimum size=6mm,inner sep=4.5pt, minimum height=1pt, minimum width= 0.4cm]
\node [client] (c1) [below=2.5cm of b1]{};

\node[bubble]     (c2) [below left=-0.105cm and -0.7cm of c1]   {};   
\node [port]      (c3) [below left=-0.58cm and -1.1cm of c1]  {$p_{n_1}$}; 
 
\node []  (k14)  [above=0.4cm of c2]  {$k_{n_1}$}; 

\tikzstyle{port}=[rectangle,ultra thin,draw=black!75,minimum size=6mm,inner sep=4.5pt, minimum height=2pt, minimum width= 0.4cm]
\node[bubble]     (c4) [below=-0.105cm of c1]   {};  
\node [port]      (c5) [below=-0.60cm and -2.0cm of c1]  {$p_{q_1}$}; 

\node []  (k15)  [above=0.4cm of c4]  {$k_{q_1}$}; 

\tikzstyle{port}=[rectangle,ultra thin,draw=black!75,minimum size=6mm,inner sep=4.5pt, minimum height=1pt, minimum width= 0.4cm]
\node[bubble]     (c6) [below right=-0.105cm and -0.75cm of c1]   {};   
\node [port]      (c7) [below right=-0.58cm and -1.05cm of c1]  {$p_{c_1}$};  

\node []  (k16)  [above=0.4cm of c6]  {$k_{c_1}$}; 

\tikzstyle{port}=[rectangle,ultra thin,draw=black!75,minimum size=6mm,inner sep=4.5pt, minimum height=2.4pt, minimum width=0.8cm]
\node[bubble] (c8)  [above right=-0.6 and -0.1cm of c1]   {};   
\node [port]  (c9)  [above right=-0.8cm and -.81cm of c1]  {$p_{l_1}$};

\node []  (k17)  [left=0.6cm of c8]  {$k_{l_1}$};

\node[bubble] (c10)  [below right=-1.6cm and -0.1cm of c1]   {};   
\node [port]  (c11)  [below right=-1.8 cm and -.82cm of c1]  {$p_{o_1}$};
 
\node []  (k18)  [left=0.6cm of c10]  {$k_{o_1}$};

\node[type]    (c12)        [above left=-0.8cm and -2.3cm of c1]{Client $1$};
\node                       [below=-0.1cm of c12]{$wC(1)$};

\tikzstyle{port}=[rectangle,ultra thin,draw=black!75,minimum size=6mm,inner sep=4.5pt, minimum height=1pt, minimum width= 0.4cm]
\node [component] (d1) [right=9cm of b1]{};
\node[bubble] (d2) [below left=-0.105cm and -1.2cm of d1]   {};   
\node [port] (d3) [below left=-0.60cm and -1.5cm of d1]  {$p_{g_2}$};  

\node []  (k19)  [above=0.4cm of d2]  {$k_{g_2}$};

\node[bubble] (d4) [below right=-0.105cm and -1.2cm of d1]   {};   
\node [port] (d5) [below right=-0.58cm and -1.5cm of d1]  {$p_{s_2}$};  

\node []  (k20)  [above=0.4cm of d4]  {$k_{s_2}$};

\node[bubble] (d6) [left=-0.1cm of d1]   {};   
\node [port] (d7) [left=-.81cm of d1]  {$p_{r_2}$}; 

\node []  (k20)  [right=0.6cm of d6]  {$k_{r_2}$};

\node[type] (d8) [above=-0.8cm of d1]{Service $2$};
\node            [below=-0.1cm of d8]{$wS(2)$};

\node [component] (e1) [below right=2cm and 1.1cm of d1]{};
\node[bubble] (e2) [below left=-0.105cm and -1.2cm of e1]   {};   
\node [port] (e3) [below left=-0.60cm and -1.5cm of e1]  {$p_{g_1}$};  

\node []  (k21)  [above=0.4cm of e2]  {$k_{g_1}$};

\node[bubble] (e4) [below right=-0.105cm and -1.2cm of e1]   {};   
\node [port] (e5) [below right=-0.58cm and -1.5cm of e1]  {$p_{s_1}$}; 

\node []  (k22)  [above=0.4cm of e4]  {$k_{s_1}$};

\node[bubble] (e6) [left=-0.1cm of e1]   {};   
\node [port] (e7) [left=-.81cm of e1]  {$p_{r_1}$};  

\node []  (k23)  [right=0.6cm of e6]  {$k_{r_1}$};

\node[type] (e8) [above=-0.8cm of e1]{Service $1$};
\node        [below=-0.1cm of e8]{$wS(1)$};

\path[-]          (g2)  edge                  node {} (d6);
\path[-]          (d6)  edge                  node {} (g2);

\path[-]          (g2)  edge                  node {} (e6);
\path[-]          (e6)  edge                  node {} (g2);
\path[-]          (g4)  edge                  node {} (c8);
\path[-]          (c8)  edge                  node {} (g4);

\path[-]          (g6)  edge                  node {} (c10);
\path[-]          (c10) edge                  node {} (g6);

\path[-]          (c2)  edge                  node {} (h2);
\path[-]          (c4)  edge                  node {} (h4);
\path[-]          (c6)  edge                  node {} (h6);

\path[-]          (h2)  edge                  node {} (c2);
\path[-]          (h4)  edge                  node {} (c4);
\path[-]          (h6)  edge                  node {} (c6);


\path[-]          (g4)  edge                  node {} (b8);
\path[-]          (b8)  edge                  node {} (g4);

\path[-]          (g6)  edge                  node {} (b10);
\path[-]          (b10)  edge                  node {} (g6);

\path[-]          (b2)  edge                  node {} (h2);
\path[-]          (h2)  edge                  node {} (b2);

\path[-]          (b4)  edge                  node {} (h4);
\path[-]          (h4)  edge                  node {} (b4);

\path[-]          (b6)  edge                  node {} (h6);
\path[-]          (h6)  edge                  node {} (b6);

\path[-]          (d2)  edge                  node {} (h4);
\path[-]          (h4)  edge                  node {} (d2);

\path[-]          (d4)  edge                  node {} (h6);
\path[-]          (h6)  edge                  node {} (d4);

\end{tikzpicture}}
\caption{A possible execution of the interactions in a weighted Request/Response architecture.}
\label{b_r_r}
\end{figure}
\end{example}

\begin{example} \textbf{\emph{(Weighted Publish/Subscribe)}}
\label{b-pu-su}
Publish/Subscribe architecture is widely used in IoT 
 applications (cf. for instance \cite{Ol:AP,Pa:Pu}), and recently in cloud systems \cite{Ya:Pr} and robotics \cite{Ma:Ho}. Publish/Subscribe architecture involves
 three types of components, namely publishers, subscribers, and topics (Figure \ref{b-p-s}). 
 Publishers advertise and transmit to
topics the type of messages they are able to produce. Then, subscribers are connected with topics
they are interested in,  and topics in turn transfer the messages from publishers to corresponding subscribers. 
Once a subscriber receives
the message it has requested, then it is disconnected from the relevant topic. Publishers cannot check the existence of subscribers and vice-versa \cite{Eg:Pu}.

Let $(w\B,\widetilde{\varphi})$ be a weighted component-based system with the Publish/Subscribe architecture. For our example we consider two weighted publisher components, two weighted topic components and three weighted subscriber components. 
Hence, $w\B=\lbrace wB(1),wB(2),wT(1),wT(2),\\wS(1),wS(2),wS(3)\rbrace$ refer to the aforementioned components, respectively. The corresponding sets of ports are $\lbrace p_{a_1},p_{t_1}\rbrace,\lbrace p_{a_2},p_{t_2}\rbrace, \lbrace p_{n_1},p_{r_1},p_{c_1},p_{s_1},p_{f_1}\rbrace, \lbrace p_{n_2},p_{r_2},p_{c_2},p_{s_2},p_{f_2}\rbrace,$ $\lbrace p_{e_1},p_{g_1}, \\ p_{d_1} \rbrace$, $\lbrace p_{e_2},p_{g_2},p_{d_2} \rbrace$, and $\lbrace p_{e_3},p_{g_3},p_{d_3} \rbrace$, respectively. 
 Figure \ref{b-p-s} depicts one of the possible instantiations for the interactions  
 among the components of our system. The ports $p_{a_z}$ and $p_{t_z}$, for $z=1,2$, are used from publishers for advertising and transferring their
 messages to topic components, respectively. Each of the two topics is notified from the publishers and receives their messages through ports $p_{n_z}$ and $p_{r_z}$, for $z=1,2$,
 respectively. Ports $p_{c_z},p_{s_z}$ and $p_{f_z}$, for $z=1,2$, are used from  topic components for the connection with a subscriber, the sending of 
 a message to a subscriber and for finalizing their connection (disconnection), respectively. Subscribers use the ports $p_{e_y},p_{g_y},p_{d_y} $, for $y=1,2,3$, for connecting with the topic (express interest),
getting a message from the topic, and disconnecting from the topic, respectively. The interactions in the architecture range over $I_{\B}$ and the weight of each port is shown in Figure \ref{b-p-s}. The \emph{wEPIL} formula $\e$ for the Publish/Subscribe architecture is $\e=\e_1\oplus \e_2\oplus (\e_1\varpi \e_2)$ with
\begin{multline*}
\e_1=\bigg(\big( \x_1 \odot \e_{11}\big)\oplus \big( \x_1 \odot \e_{12}\big)\oplus \big(  \x_1 \odot \e_{13}\big)\oplus \\
\big(  \x_1 \odot (\e_{11}\varpi \e_{12})\big)\oplus \big( \x_1 \odot (\e_{11}\varpi \e_{13})\big)\oplus \\ \big( \x_1 \odot (\e_{12}\varpi \e_{13})\big)\oplus \big( \x_1 \odot (\e_{11}\varpi\e_{12}\varpi \e_{13})\big)\bigg)
\end{multline*}
and 
\begin{multline*}
\e_2=\bigg(\big(  \x_2 \odot \e_{21}\big)\oplus \big( \x_2 \odot \e_{22}\big)\oplus \big(  \x_2 \odot \e_{23}\big)\oplus \\
\big(  \x_2 \odot (\e_{21}\varpi \e_{22})\big)\oplus \big( \x_2 \odot (\e_{21}\varpi \e_{23})\big)\oplus \\ \big( \x_2 \odot (\e_{22}\varpi \e_{23})\big)\oplus \big(  \x_2 \odot (\e_{21}\varpi \e_{22}\varpi \e_{23})\big)\bigg)
\end{multline*}
\noindent where  the following auxiliary subformulas: 
\begin{itemize}
    \item[-] $\x_1=\x_{11}\oplus \x_{12}\oplus (\x_{11}\varpi \x_{12})$
    \item[-] $\x_2=\x_{21}\oplus \x_{22}\oplus (\x_{21}\varpi \x_{22})$
\end{itemize}
\noindent represent the cost for the connection of each of the two topics with the first publisher, or
the second one or with both of them, and
\begin{itemize}
\item[-] $\x_{11}=\#_w(p_{n_1}\otimes p_{a_1}) \odot \#_w(p_{r_1}\otimes p_{t_1})$
\item[-] $\x_{12}=\#_w(p_{n_1}\otimes p_{a_2}) \odot \#_w(p_{r_1}\otimes p_{t_2})$
\item[-] $\x_{21}=\#_w(p_{n_2}\otimes p_{a_1}) \odot \#_w(p_{r_2}\otimes p_{t_1})$
\item[-] $\x_{22}=\#_w(p_{n_2}\otimes p_{a_2}) \odot \#_w(p_{r_2}\otimes p_{t_2})$
\end{itemize}
\noindent describe the cost of the interactions of the two topics with each of the two publishers,  and 
\begin{itemize}
\item[-] $\e_{11}= \#_w(p_{c_1}\otimes p_{e_1}) \odot \#_w(p_{s_1}\otimes p_{g_1}) \odot \#_w(p_{f_1}\otimes p_{d_1})$
\item[-] $\e_{12}= \#_w(p_{c_1}\otimes p_{e_2}) \odot \#_w(p_{s_1}\otimes p_{g_2}) \odot \#_w(p_{f_1}\otimes p_{d_2})$
\item[-] $\e_{13}= \#_w(p_{c_1}\otimes p_{e_3})\odot \#_w(p_{s_1}\otimes p_{g_3}) \odot \#_w(p_{f_1}\otimes p_{d_3})$
\item[-] $\e_{21}= \#_w(p_{c_2}\otimes p_{e_1}) \odot \#_w(p_{s_2}\otimes p_{g_1}) \odot \#_w(p_{f_2}\otimes p_{d_1})$
\item[-] $\e_{22}= \#_w(p_{c_2}\otimes p_{e_2})\odot \#_w(p_{s_2}\otimes p_{g_2}) \odot \#_w(p_{f_2}\otimes p_{d_2})$
\item[-] $\e_{23}= \#_w(p_{c_2}\otimes p_{e_3})\odot \#_w(p_{s_2}\otimes p_{g_3}) \odot \#_w(p_{f_2}\odot p_{d_3})$.
\end{itemize}
\noindent describe the cost of the connections of each of the three subscribers with the two topics.

\noindent Consider $w_1=\lbrace p_{a_1},p_{n_1}\rbrace \lbrace p_{t_1},p_{r_1}\rbrace \lbrace p_{c_1},p_{e_1}\rbrace\lbrace p_{s_1},p_{g_1}\rbrace \lbrace p_{c_1},p_{e_3}\rbrace \lbrace p_{f_1},p_{d_1}\rbrace \lbrace p_{s_1},p_{g_3}\rbrace \lbrace p_{f_1},p_{d_3}\rbrace$ and $w_2=\lbrace p_{a_1},p_{n_1}\rbrace \lbrace p_{t_1},p_{r_1}\rbrace \lbrace p_{c_1},p_{e_3}\rbrace\lbrace p_{c_1},p_{e_1}\rbrace \lbrace p_{s_1},p_{g_1}\rbrace \lbrace p_{c_1},p_{e_2}\rbrace\lbrace p_{s_1},p_{g_2}\rbrace\lbrace p_{s_1},p_{g_3}\rbrace \lbrace p_{f_1},p_{d_3}\rbrace \lbrace p_{f_1},\\p_{d_1}\rbrace \lbrace p_{f_1},p_{d_2}\rbrace$ where $w_1$ expresses one of the possible executions for the interactions in which subscribers $wS(1)$ and $wS(2)$ are interested in topic $wT(1)$, and $w_2$ encodes a possible implementation for applying the connections of all subscribers to $wT(1)$. For instance, in
the Viterbi semiring the value $\max \lbrace \left \Vert \e \right \Vert(w_1),\left \Vert \e \right \Vert(w_2)\rbrace$ shows the interaction executed with the maximum probability.

\begin{figure}[h]
\centering
\resizebox{0.7\linewidth}{!}{
\begin{tikzpicture}[>=stealth',shorten >=1pt,auto,node distance=1cm,baseline=(current bounding box.north)]
\tikzstyle{component}=[rectangle,ultra thin,draw=black!75,align=center,inner sep=9pt,minimum size=1.5cm,minimum height=3.9cm,minimum width=3.3cm]
\tikzstyle{port}=[rectangle,ultra thin,draw=black!75,minimum size=7mm]
\tikzstyle{bubble} = [fill,shape=circle,minimum size=5pt,inner sep=0pt]
\tikzstyle{type} = [draw=none,fill=none]

\node [component,align=center] (a1)  {};
\node [port] (a2) [above right=-1.66 and -0.78cm of a1]  {$p_{a_1}$};
\node[bubble] (a3) [below right=-0.40 and -0.10cm of a2]   {};

\node []  (k1)  [left=0.5cm of a3]  {$k_{a_1}$}; 

\node [port] (a4) [below right=0.86cm and -0.73cm of a2]  {$p_{t_1}$};
\node[bubble] (a5) [below right=-0.37 and -0.10cm of a4]   {};

\node []  (k2)  [left=0.5cm of a5]  {$k_{t_1}$};

\node[type] (a6) [above=-0.45cm of a1]{Publ. $1$};
\node[type]  [below=-0.05cm of a6]{$wB(1)$};

\node [component] (b1) [right=3cm of a1] {};
\node [port] (b2) [below right=-1.07cm and -0.77cm of b1]  {$p_{f_1}$}; 
\node[bubble] (b3) [below right=-0.40cm and -0.10cm of b2]   {};   

\node []  (k3)  [left=0.5cm of b3]  {$k_{f_1}$}; 

\node [port] (b4) [above right=0.25cm and -0.75cm of b2]  {$p_{s_1}$}; 
\node[bubble] (b5) [above right=-0.40cm and -0.10cm of b4]   {};   

\node []  (k4)  [left=0.5cm of b5]  {$k_{s_1}$}; 

\node [port] (b6) [above =0.3cm of b4]  {$p_{c_1}$}; 
\node[bubble] (b7) [above right=-0.40cm and -0.10cm of b6]   {};   

\node []  (k5)  [left=0.5cm of b7]  {$k_{c_1}$}; 

\node [port] (b8) [above left= -1.7cm and -0.79cm of b1]  {$p_{n_1}$};
\node[bubble] (b9) [above left=-0.4cm and -0.10cm of b8]   {};

\node []  (k6)  [right=0.6cm of b9]  {$k_{n_1}$}; 

\node [port] (b10) [below left=0.85 and -0.76cm of b8]  {$p_{r_1}$};
\node[bubble] (b11) [below left=-0.40 and -0.10cm of b10]   {};
 
 \node []  (k7)  [right=0.6cm of b11]  {$k_{r_1}$}; 
 
\node[type] (b12) [above left=-0.55cm and -2.4cm of b1]{Topic $1$};
\node[type]  [below=-0.1cm of b12]{$wT(1)$};

\node [component] (c1) [above right= -1.6cm and 3cm of b1] {};
\node [port] (c2) [below left=-1.07cm and -0.77cm of c1]  {$p_{d_1}$}; 
\node[bubble] (c3) [below left=-0.40cm and -0.10cm of c2]   {};  

 \node []  (k8)  [right=0.6cm of c3]  {$k_{d_1}$}; 

\node [port] (c4) [above=0.25cm and -0.75cm of c2]  {$p_{g_1}$}; 
\node[bubble] (c5) [above left=-0.40cm and -0.10cm of c4]   {};   

 \node []  (k9)  [right=0.6cm of c5]  {$k_{g_1}$}; 

\node [port] (c6) [above=0.3cm of c4]  {$p_{e_1}$}; 
\node[bubble] (c7) [above left=-0.40cm and -0.10cm of c6]   {};  

 \node []  (k10)  [right=0.6cm of c7]  {$k_{e_1}$}; 
 
\node[type] (c8) [above right=-0.55cm and -2.1cm of c1]{Subs. $1$};
\node[type]  [below=-0.1cm of c8]{$wS(1)$};

\node [component] (d1) [below right=-0.8cm and 3cm of b1] {};
\node [port] (d2) [below left=-1.07cm and -0.77cm of d1]  {$p_{d_2}$}; 
\node[bubble] (d3) [below left=-0.40cm and -0.10cm of d2]   {};  

 \node []  (k11)  [right=0.6cm of d3]  {$k_{d_2}$}; 

\node [port] (d4) [above=0.25cm and -0.75cm of d2]  {$p_{g_2}$}; 
\node[bubble] (d5) [above left=-0.40cm and -0.10cm of d4]   {};   

\node []  (k12)  [right=0.6cm of d5]  {$k_{g_2}$}; 

\node [port] (d6) [above=0.3cm of d4]  {$p_{e_2}$}; 
\node[bubble] (d7) [above left=-0.40cm and -0.10cm of d6]   {};   

 \node []  (k13)  [right=0.6cm of d7]  {$k_{e_2}$}; 

\node[type] (d8) [above right=-0.55cm and -2.1cm of d1]{Subs. $2$};
\node[type]  [below=-0.1cm of d8]{$wS(2)$};


\node [component,align=center] (e1) [below=4.5cm of a1] {};
\node [port] (e2) [above right=-1.66 and -0.78cm of e1]  {$p_{a_2}$};
\node[bubble] (e3) [below right=-0.40 and -0.10cm of e2]   {};

 \node []  (k14)  [left=0.55cm of e3]  {$k_{a_2}$}; 

\node [port] (e4) [below right=0.86cm and -0.73cm of e2]  {$p_{t_2}$};
\node[bubble] (e5) [below right=-0.37 and -0.10cm of e4]   {};

 \node []  (k15)  [left=0.5cm of e5]  {$k_{t_2}$}; 

\node[type] (e6) [above=-0.45cm of e1]{Publ. $2$};
\node[type]  [below=-0.05cm of e6]{$wB(2)$};

2nd topic

\node [component] (f1) [right=3cm of e1] {};
\node [port] (f2) [below right=-1.07cm and -0.77cm of f1]  {$p_{f_2}$}; 
\node[bubble] (f3) [below right=-0.40cm and -0.10cm of f2]   {};   

 \node []  (k16)  [left=0.5cm of f3]  {$k_{f_2}$}; 

\node [port] (f4) [above right=0.25cm and -0.75cm of f2]  {$p_{s_2}$}; 
\node[bubble] (f5) [above right=-0.40cm and -0.10cm of f4]   {};   

 \node []  (k17)  [left=0.5cm of f5]  {$k_{s_2}$}; 

\node [port] (f6) [above =0.3cm of f4]  {$p_{c_2}$}; 
\node[bubble] (f7) [above right=-0.40cm and -0.10cm of f6]   {};   

 \node []  (k18)  [left=0.5cm of f7]  {$k_{c_2}$}; 

\node [port] (f8) [above left= -1.7cm and -0.79cm of f1]  {$p_{n_2}$};
\node[bubble] (f9) [above left=-0.4cm and -0.10cm of f8]   {};

 \node []  (k19)  [right=0.6cm of f9]  {$k_{n_2}$}; 

\node [port] (f10) [below left= 0.85cm and -0.76cm of f8]  {$p_{r_2}$};
\node[bubble] (f11) [below left=-0.40 and -0.10cm of f10]   {};
 
  \node []  (k20)  [right=0.6cm of f11]  {$k_{r_2}$}; 
 
\node[type] (f12) [above left=-0.55cm and -2.4cm of f1]{Topic $2$};
\node[type]  [below=-0.1cm of f12]{$wT(2)$};

\node [component] (g1) [right= 3cm of f1] {};
\node [port] (g2) [below left=-1.07cm and -0.77cm of g1]  {$p_{d_3}$}; 
\node[bubble] (g3) [below left=-0.40cm and -0.10cm of g2]   {};  

 \node []  (k21)  [right=0.6cm of g3]  {$k_{d_3}$}; 

\node [port] (g4) [above left=0.25cm and -0.75cm of g2]  {$p_{g_3}$}; 
\node[bubble] (g5) [above left=-0.40cm and -0.10cm of g4]   {};   

 \node []  (k22)  [right=0.6cm of g5]  {$k_{g_3}$}; 

\node [port] (g6) [above=0.3cm of g4]  {$p_{e_3}$}; 
\node[bubble] (g7) [above left=-0.40cm and -0.10cm of g6]   {};  

 \node []  (k23)  [right=0.6cm of g7]  {$k_{e_3}$}; 
 
\node[type] (g8) [above right=-0.55cm and -2.1cm of g1]{Subs. $3$};
\node[type]  [below=-0.1cm of g8]{$wS(3)$};

\path[-]          (a3)  edge                  node {} (b9);
\path[-]          (b9)  edge                  node {} (a3);

\path[-]          (a5)  edge                  node {} (b11);
\path[-]          (b11)  edge                  node {} (a5);

\path[-]          (b7)  edge                  node {} (c7);
\path[-]          (c7)  edge                  node {} (b7);

\path[-]          (b7)  edge                  node {} (g7);
\path[-]          (g7)  edge                  node {} (b7);

\path[-]          (b5)  edge                  node {} (c5);
\path[-]          (c5)  edge                  node {} (b5);

\path[-]          (b5)  edge                  node {} (g5);
\path[-]          (g5)  edge                  node {} (b5);

\path[-]          (b3)  edge                  node {} (c3);
\path[-]          (c3)  edge                  node {} (b3);

\path[-]          (b3)  edge                  node {} (g3);
\path[-]          (g3)  edge                  node {} (b3);


\path[-]          (e3)  edge                  node {} (f9);
\path[-]          (f9)  edge                  node {} (e3);

\path[-]          (e5)  edge                  node {} (f11);
\path[-]          (f11)  edge                  node {} (e5);

\path[-]          (a3)  edge                  node {} (f9);
\path[-]          (f9)  edge                  node {} (a3);

\path[-]          (a5)  edge                  node {} (f11);
\path[-]          (f11)  edge                  node {} (a5);


\path[-]          (f7)  edge                  node {} (d7);
\path[-]          (d7)  edge                  node {} (f7);

\path[-]          (f5)  edge                  node {} (d5);
\path[-]          (d5)  edge                  node {} (f5);

\path[-]          (f3)  edge                  node {} (d3);
\path[-]          (d3)  edge                  node {} (f3);

\end{tikzpicture}}
\caption{A possible execution of the interactions in a weighted Publish/Subscribe architecture.}
\label{b-p-s}
\end{figure}
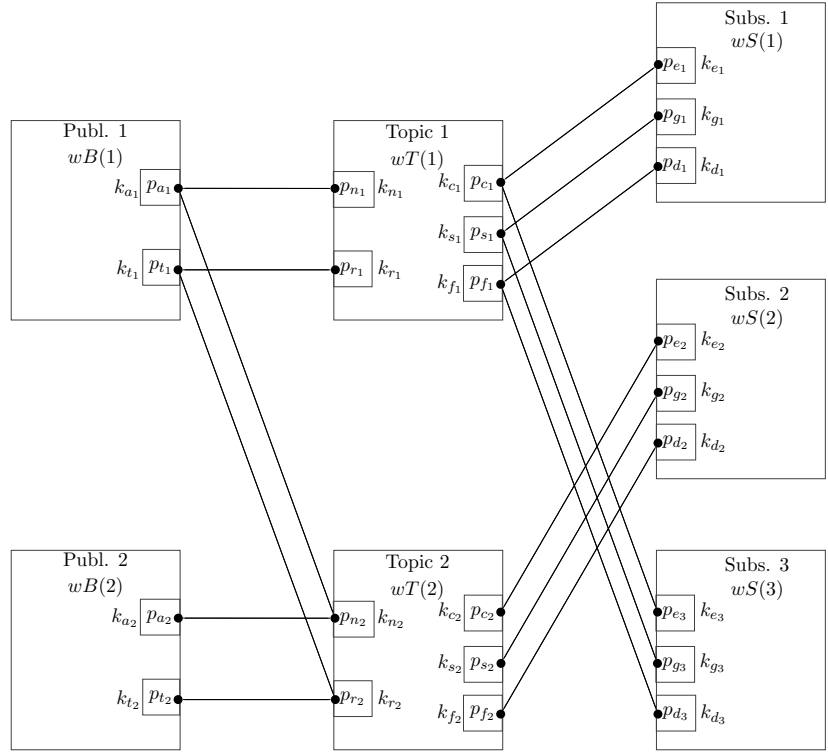

\end{example}

The above examples demonstrate that wEPIL formulas can efficiently represent the overall cost of ordered interactions within architectures. The resulting value is obtained by the relevant operations of each semiring and the notion of cost serves for representing different quantitative values like the required power consumption, battery waste, or time, for implementing
the components' connections.

\section{Parametric weighted component-based systems}

In this section we deal with the parametric extension of weighted component-based systems. 
Weighted component-based systems considered in Subsection \ref{comp_based_system} are comprised of a finite number of weighted components 
which are of the same or distinct type. On the other hand, in the parametric setting a
weighted component-based model is comprised of a finite number of distinct \emph{weighted component types} where the cardinality 
of the \emph{instances} of each type is a parameter for the system.  It should be clear, that in real world applications we do not need an unbounded number of components. Nevertheless, the number of instances of every component type is unknown or it can be modified during a process. Therefore, in the
sequel we consider parametric weighted component-based systems, i.e., weighted component-based systems with
infinitely many instances of every component type. 
We shall need to recall firstly parametric (unweighted) component-based systems \cite{Pi:Ar}. 

Let $\B=\{B(i) \mid i \in [n]  \}$ be a set of component types. For every $i \in [n] $ and $j \geq 1$ we consider a copy $B(i,j)=(Q(i,j),P(i,j),q_{0}(i,j),R(i,j))$ of $B(i)$, namely the \emph{$j$-th instance of} $B(i)$.  Hence, for every $i \in [n]$ and $j \geq 1$, the instance $B(i,j)$ is also a component and we call it a \emph{parametric component} or a \emph{component instance}. We assume that $(Q(i,j) \cup P(i,j)) \cap (Q(i', j') \cup  P(i',j')) = \emptyset$ whenever $i \neq i' $ or $j \neq j'$ for every $ i, i' \in [n]$ and $j, j' \geq 1$.  This restriction is needed in order to identify the distinct parametric components. It also permits us to use, without any confusion, the notation  $P(i,j)=\{p(j) \mid p\in P(i)\}$ for every $i \in [n]$ and $j \geq 1$.   We set $p\mathcal{B}= \{B(i,j) \mid i \in [n],  j \geq 1 \} $ and call it a set of \emph{parametric components}. The set of ports of $p\B$ is given by $P_{p\B} =\bigcup_{i \in [n], j\geq 1}  P(i,j)$.

Let $w\B=\{wB(i) \mid i \in [n]  \}$ be a set of weighed component types. For every $i \in [n] $ and $j \geq 1$ we consider a weighted component   $wB(i,j)=(B(i,j), wt(i))$, where $B(i,j)=(Q(i,j),P(i,j),q_{0}(i,j),R(i,j))$ is the $j$-th instance of $B(i)$, and it is called a \emph{parametric weighted component} or a \emph{weighted component instance}. We set $pw\mathcal{B}= \{wB(i,j) \mid i \in [n],  j \geq 1 \} $ and call it a set of \emph{parametric weighted components}. We impose on $pw\B$ the same  assumptions as for $p\B$. Abusing notations, we denote by $wt(i)$, $i \in [n]$, the weight mapping of $wB(i,j)$, $j\geq 1$, meaning that it assigns the value $wt(i)(p)$ to every port $p(j)\in P(i,j)$. 

As it is already mentioned, in practical applications we do not know how many instances of each weighted component type are connected at a concrete time. This means that we cannot define interactions of $pw\mathcal{B}$ in the same way as we did  for finite sets of weighted component types. For this, we need a symbolic representation to describe interactions and hence architectures of parametric weighted systems. In \cite{Pi:Ar} we investigated the  first-order extended interaction
 logic (FOEIL for short) which was proved sufficient to describe a wide class of architectures of parametric component-based systems. Here, we introduce a weighted first-order extended interaction logic for the description of architectures of parametric weighted component-based systems.

\subsection{Weighted first-order extended interaction logic}
\label{sub_wFOEIL}
In this subsection we introduce the  weighted first-order extended interaction logic as a modelling language for describing the interactions of parametric weighted component-based systems. 

As in FOEIL  we  equip wEPIL formulas with variables. Due to the
nature of parametric systems we need to distinguish variables referring to different component  types. Let $pw\mathcal{B}= \{wB(i,j) \mid i \in [n], j \geq 1 \} $ be a set of parametric weighted components. We consider pairwise disjoint countable sets of first-order variables  $\mathcal{X}^{(1)}, \ldots , \mathcal{X}^{(n)}$ referring to  instances of component types $wB(1), \ldots, wB(n)$, respectively. Variables in $\mathcal{X}^{(i)}$, for every $i \in [n]$, will be denoted by small  letters with the corresponding superscript. Hence by $x^{(i)} $  we understand that $x^{(i)} \in \mathcal{X}^{(i)}$, $i\in [n]$, is a first-order variable referring to an  instance of weighted component type $wB(i)$. We let $\mathcal{X}=\mathcal{X}^{(1)}\cup \ldots \cup \mathcal{X}^{(n)}$ and  set  $P_{p\B(\mathcal{X})} =\left\{p\left(x^{(i)}\right) \mid i \in [n], x^{(i)} \in \mathcal{X}^{(i)},  \text{ and } p\in P(i)\right \}$.
 
\noindent We need to recall FOEIL firstly \cite{Pi:Ar}.  
Let $p\mathcal{B}= \{B(i,j) \mid i \in [n], j \geq 1 \} $ be a set of parametric components. Then the syntax of FOEIL formulas $\psi$ over $p\mathcal{B}$\footnote{According to our terminology for EPIL formulas, a FOEIL formula should be defined over the set of ports of $p\B$. Nevertheless, we prefer for simplicity to refer to the set $p\B$ of parametric components.} is given by the grammar
\begin{multline*}
\psi  ::= \varphi  \mid x^{(i)}=y^{(i)} \mid  \neg(x^{(i)}=y^{(i)}) \mid  \psi\vee\psi \mid \psi \wedge \psi \mid   \psi * \psi   \mid \psi\shuffle \psi \mid \\ \exists x^{(i)}. \psi \mid  \forall x^{(i)} . \psi \mid \exists^* x^{(i)}. \psi \mid \forall^*x^{(i)}.\psi \mid \exists^{\shuffle} x^{(i)}.\psi \mid \forall^{\shuffle}x^{(i)}.\psi  
\end{multline*}
where $\varphi$ is an EPIL formula over $P_{p\B(\mathcal{X})}$, $i \in [n]$,  $x^{(i)}, y^{(i)} \in \mathcal{X}^{(i)}$, $\exists^*$ (resp. $\forall^*$) denotes the existential  (resp. universal) concatenation quantifier, and $\exists ^{\shuffle}$ (resp. $\forall^{\shuffle}$) the existential (resp. universal) shuffle quantifier. Furthermore, we assume that whenever $\psi$ contains a subformula of the form $\exists ^*x^{(i)}.\psi'$ or $\exists ^{\shuffle} x^{(i)}.\psi'$,  then the application of negation in $\psi'$ is permitted only  in PIL formulas and formulas of the form $x^{(j)}=y^{(j)}$.

Let $\psi$ be a FOEIL formula over $p\B$. We denote by $\mathrm{free}(\psi)$ the set of free variables of $\psi$. If $\psi$ has no free variables, then it is a \emph{sentence}. 
We consider a mapping $r:[n] \rightarrow \mathbb{N}$. The value $r(i)$, for every $i \in [n]$, intends to represent the finite number of instances of the component type $B(i)$ in the parametric system. The mapping characterizes the dynamic behavior of such systems, where components' instances can appear or disappear, affecting in turn, the corresponding interactions. Hence, for different mappings we obtain a different parametric system. 
We let $p\B(r)= \{B(i,j) \mid i \in [n], j \in [r(i)] \} $ and call it  the \emph{instantiation of} $p\B$ w.r.t. $r$. We denote by $P_{p\B(r)}$ the set of all ports of components' instances in $p\B(r)$, i.e., $P_{p\B(r)}= \bigcup_{i\in [n], j \in [r(i)]}P(i,j)$. The set $I_{p\B(r)}$ of interactions of $p\B(r)$ is given by $I_{p\B(r)}=\{a \in I(P_{p\B(r)}) \mid  \vert a \cap P(i,j)\vert \leq 1 \text{ for every } i\in [n] \text{ and } j \in [r(i)]\}$.

Let $\mathcal{V} \subseteq \mathcal{X} $ be a finite set of first-order variables.  We let $P_{p\B(\mathcal{V})}= \{ p(x^{(i)}) \in P_{p\B(\mathcal{X})} \mid x^{(i)} \in \mathcal{V} \}$.
To interpret FOEIL formulas over $p\B$ we use the notion of an \emph{assignment} defined with respect to the set of variables $\mathcal{V}$ and the mapping $r$.  \hide{We let $\mathcal{A}_{r}=\lbrace (i,j)\mid i\in [n],j\in [r(i)]\rbrace$.} Formally, a $(\mathcal{V},r)$-\emph{assignment} is a mapping $\sigma : \mathcal{V} \rightarrow \mathbb{N}  $ such that $\sigma(\mathcal{V} \cap \mathcal{X}^{(i)} ) \subseteq [r(i)]$ for every $i \in [n]$. If $\sigma$ is a $(\mathcal{V},r)$-assignment, then $\sigma[x^{(i)} \rightarrow j]$  is the $(\mathcal{V} \cup \{x^{(i)}\},r)$-assignment which acts as $\sigma$ on $\mathcal{V}\setminus \{x^{(i)}\}$  and assigns $j$ to $x^{(i)}$. Intuitively, a $(\mathcal{V},r)$-assignment $\sigma$ assigns unique identifiers to each instance in a parametric system, w.r.t. the mapping $r$. 

We interpret FOEIL formulas over triples consisting of a mapping $r:[n] \rightarrow \mathbb{N}$, a $(\mathcal{V},r)$-assignment $\sigma$, and a word $w \in I_{p\B(r)}^*$.

\begin{definition}
Let $\psi$ be a \emph{FOEIL} formula over a set $p\mathcal{B}= \{B(i,j) \mid i \in [n], j \geq 1 \} $ of parametric components and $\mathcal{V} \subseteq \mathcal{X}$ a finite set containing $\mathrm{free}(\psi)$. Then for every $r:[n] \rightarrow \mathbb{N}$, $(\mathcal{V},r)$-assignment $\sigma$, and $w  \in I_{p\B(r)}^*$ we define the satisfaction relation $(r,\sigma,w) \models\psi$, inductively on the structure of $\psi$ as follows:

\begin{itemize}

     \item[-] $(r,\sigma,w) \models  \varphi$ iff $w \models \sigma(\varphi)$,
     
     \item[-] $(r,\sigma,w) \models  x^{(i)}=y^{(i)}$ iff $\sigma(x^{(i)})=\sigma(y^{(i)})$,
     
     \item[-] $(r,\sigma,w) \models  \neg(x^{(i)}=y^{(i)})$ iff $(r,\sigma,w) \not\models x^{(i)}=y^{(i)}$,

     \item[-]  $(r,\sigma, w)  \models \psi_1 \vee \psi_2$ iff $(r,\sigma, w)\models  \psi_1$ or $(r,\sigma, w)\models \psi_2$, 
 
\item[-]  $(r,\sigma, w)  \models \psi_1 \wedge \psi_2$ iff $(r,\sigma, w)\models  \psi_1$ and $(r,\sigma, w)\models \psi_2$,  
     
     \item[-] $(r,\sigma, w) \models \psi_1 * \psi_2$ iff there exist $w_1,w_2 \in I_{p\B(r)}^*$ such that $w=w_1w_2$ and   $(r,\sigma, w_i) \models  \psi_i$ for $i=1,2$, 
     \item[-] $(r, \sigma, w) \models \psi_1 \shuffle \psi_2$ iff there exist $w_1, w_2 \in I_{p\B(r)}^*$ such that $w \in w_1\shuffle w_2$ and $(r, \sigma, w_i) \models \psi_i$ for $i=1,2$, 
     
     \item[-]   $(r,\sigma, w) \models \exists x^{(i)} . \psi$ iff there exists $j \in [r(i)] $ such that  $(r, \sigma[x^{(i)} \rightarrow  j ], w) \models \psi$,

      \item[-]  $(r,\sigma, w) \models \forall x^{(i)} . \psi$ iff  $(r, \sigma[x^{(i)} \rightarrow j ], w) \models \psi$ for every $j \in [r(i)]$,
      
\item[-]  $(r,\sigma, w) \models \exists^ * x^{(i)} . \psi$ iff there exist $w_{l_1}, \ldots, w_{l_t} \in I_{p\B(r)}^*$ with $1 \leq l_1 < \ldots < l_t \leq r(i)$ such that $w = w_{l_1}  \ldots  w_{l_t}$ and $(r, \sigma[x^{(i)} \rightarrow j ], w_j) \models \psi$ for every $j = l_1, \ldots, l_t$,

      \item[-]  $(r,\sigma, w) \models \forall^* x^{(i)} . \psi$ iff there exist $w_1, \ldots, w_{r(i)} \in I_{p\B(r)}^*$ such that $w=w_1 \ldots w_{r(i)}$ and $(r, \sigma[x^{(i)} \rightarrow j ], w_j) \models \psi$ for every $j \in [r(i)]$,
      
\item[-]  $(r,\sigma, w) \models \exists^{\shuffle} x^{(i)} . \psi$ iff there exist $w_{l_1}, \ldots, w_{l_t} \in I_{p\B(r)}^*$ with $1 \leq l_1 < \ldots < l_t \leq r(i)$ such that $w \in w_{l_1} \shuffle \ldots \shuffle w_{l_t}$ and $(r, \sigma[x^{(i)} \rightarrow j ], w_j) \models \psi$ for every $j = l_1, \ldots, l_t$,

\item[-]  $(r,\sigma, w) \models \forall^{\shuffle} x^{(i)} . \psi$ iff there exist $w_1, \ldots, w_{r(i)} \in I_{p\B(r)}^*$ such that $w \in w_1 \shuffle \ldots \shuffle w_{r(i)}$ and $(r, \sigma[x^{(i)} \rightarrow j ], w_j) \models \psi$ for every $j \in [r(i)]$,

\end{itemize}
where $\sigma(\varphi)$ is obtained by $\varphi$ by replacing every port $p(x^{(i)}) \in P_{p\B(\mathcal{V})}$, occurring in $\varphi$, by $p(\sigma(x^{(i)}))$.
\end{definition}

If $\psi$ is a FOEIL sentence over $p\B$, then we simply write $(r, w) \models \psi$. Let also $\psi'$ be a FOEIL sentence over $p\mathcal{B}$. Then,  $\psi$ and $\psi'$ are called \emph{equivalent w.r.t.} $r$ whenever $(r,w) \models \psi$ iff $(r,w) \models \psi'$, for every $w  \in I_{p\B(r)}^*$.

In the sequel, we shall write also $x^{(i)} \neq y^{(i)}$  for $\neg(x^{(i)}=y^{(i)})$.

Let $\beta$ be a boolean combination of atomic formulas of the form  $x^{(i)}=y^{(i)}$ and $\psi$ a FOEIL formula over $p\B$. Then, we define 
$\beta \rightarrow \psi ::= \neg \beta \vee \psi$.

For simplicity sometimes we denote boolean combinations of  formulas of the form $x^{(i)} = y^{(i)}$ as constraints. For instance we write $\exists x^{(i)}\forall y^{(i)}\exists x^{(j)}\forall y^{(j)} ((x^{(i)}  \neq y^{(i)}) \wedge (x^{(j)} \neq y^{(j)})) .\psi$ for  $\exists x^{(i)}\forall y^{(i)}\exists x^{(j)}\forall y^{(j)}. (((x^{(i)}\neq y^{(i)}) \wedge (x^{(j)} \neq y^{(j)})) \rightarrow \psi)$. 

Now we are ready to introduce our weighted FOEIL.  
\begin{definition} \label{def_wfOEIL}
Let $pw\mathcal{B}= \{wB(i,j) \mid i \in [n], j \geq 1 \} $ be a set of parametric weighted components. Then the syntax of \emph{weighted first-order extended interaction logic} \emph{(wFOEIL} for short\emph{) formulas} $\ps$ \emph{over} $pw\mathcal{B}$ and $K$ is given by the grammar
\begin{multline*}
\ps  ::= k \mid \psi   \mid  \ps \oplus \ps \mid \ps \otimes \ps \mid   \ps \odot \ps   \mid \ps\varpi \ps \mid  {\textstyle\sum x^{(i)}}. \ps \mid  {\textstyle\prod x^{(i)}} . \ps \mid  \\ {\textstyle\sum\nolimits{^{\odot}}} x^{(i)}. \ps \mid {\textstyle\prod\nolimits{^{\odot}}x^{(i)}}.\ps \mid {\textstyle\sum\nolimits{^{\varpi}} x^{(i)}}.\ps \mid {\textstyle\prod\nolimits{^{\varpi}}x^{(i)}}.\ps  
\end{multline*}
where $k \in K$, $\psi$ is a \emph{FOEIL} formula over $p\B$,  $x^{(i)}, y^{(i)}$ are first-order variables in $\mathcal{X}^{(i)}$, ${\textstyle\sum}$ (resp. ${\textstyle\prod}$) denotes the weighted existential (resp. universal) quantifier, ${\textstyle\sum\nolimits{^{\odot}}}$ (resp. $ {\textstyle\prod\nolimits{^{\odot}}}$) denotes the weighted existential (resp. universal) concatenation quantifier, and ${\textstyle\sum\nolimits{^{\varpi}}}$ (resp. ${\textstyle\prod\nolimits{^{\varpi}}}$) the weighted existential (resp. universal) shuffle quantifier. Furthermore, we assume that when $\ps$ contains a subformula of the form $ {\textstyle\sum\nolimits{^{\odot}}} x^{(i)}. \ps'$ or $ {\textstyle\sum\nolimits{^{\varpi}} x^{(i)}}.\ps'$, and $\ps'$ contains a \emph{FOEIL} formula $\psi$, then the application of negation in $\psi$ is permitted only in \emph{PIL} formulas, and formulas of the form $x^{(j)}=y^{(j)}$.
\end{definition}

Let $\ps$ be a wFOEIL formula over $pw\B$ and $r:[n] \rightarrow \mathbb{N}$ a mapping. As for (unweighted) parametric systems the value $r(i)$, for every $i \in [n]$, intends to represent the finite number of instances of the weighted component type $wB(i)$ in the parametric system. 
We let $pw\B(r)= \{wB(i,j) \mid i \in [n], j \in [r(i)] \} $ and call it  the \emph{instantiation of} $pw\B$ w.r.t. $r$. The set of ports and the set of interactions of $pw\B(r)$ are the same as the corresponding ones in $p\B(r)$, hence we use for simplicity
the same symbols $P_{p\B(r)}$ and $I_{p\B(r)}$, respectively.
\hide{
Let $\mathcal{V} \subseteq \mathcal{X} $ be a finite set of first-order variables.  We let $P_{p\B(\mathcal{V})}= \{ p(x^{(i)}) \in P_{p\B(\mathcal{X})} \mid x^{(i)} \in \mathcal{V} \}$.
To interpret FOEIL formulas over $p\B$ we use the notion of an \emph{assignment} defined with respect to the set of variables $\mathcal{V}$ and the mapping $r$.  \hide{We let $\mathcal{A}_{r}=\lbrace (i,j)\mid i\in [n],j\in [r(i)]\rbrace$.} Formally, a $(\mathcal{V},r)$-\emph{assignment} is a mapping $\sigma : \mathcal{V} \rightarrow \mathbb{N}  $ such that $\sigma(\mathcal{V} \cap \mathcal{X}^{(i)} ) \subseteq [r(i)]$ for every $i \in [n]$. If $\sigma$ is a $(\mathcal{V},r)$-assignment, then $\sigma[x^{(i)} \rightarrow j]$  is the $(\mathcal{V} \cup \{x^{(i)}\},r)$-assignment which acts as $\sigma$ on $\mathcal{V}\setminus \{x^{(i)}\}$  and assigns $j$ to $x^{(i)}$. Intuitively, a $(\mathcal{V},r)$-assignment $\sigma$ assigns unique identifiers to each instance in a parametric system, w.r.t. the mapping $r$. }

We interpret wFOEIL formulas $\ps$ as series $\Vert \ps  \Vert$ over triples consisting of a mapping $r:[n] \rightarrow \mathbb{N}$, a $(\mathcal{V},r)$-assignment $\sigma$, and a word $w \in I_{p\B(r)}^*$. Intuitively, the use of weighted existential and universal concatenation (resp. shuffle) quantifiers ${\textstyle\sum\nolimits{^{\odot}}} x^{(i)}. \ps$ and ${\textstyle\prod\nolimits{^{\odot}}x^{(i)}}.\ps$ (resp. ${\textstyle\sum\nolimits{^{\varpi}} x^{(i)}}.\ps$ and ${\textstyle\prod\nolimits{^{\varpi}}x^{(i)}}.\ps $) serves to compute the weight of the partial and whole participation of the weighted component instances, determined by the application of the assignment $\sigma$ to $x^{(i)}$, in sequential (resp. interleaving) interactions.

\begin{definition}
Let $\ps$ be a \emph{wFOEIL} formula over a set $pw\mathcal{B}= \{wB(i,j) \mid i \in [n], j \geq 1 \} $ of parametric weighted components and $K$, and $\mathcal{V} \subseteq \mathcal{X}$ a finite set containing $\mathrm{free}(\psi)$. Then for every $r:[n] \rightarrow \mathbb{N}$, $(\mathcal{V},r)$-assignment $\sigma$, and $w  \in I_{p\B(r)}^*$ we define the value   $\Vert\ps\Vert(r,\sigma,w)$, inductively on the structure of $\ps$ as follows:

\begin{itemize}

     \item[-] $\Vert k \Vert(r,\sigma,w) = k$, 
     
     \item[-] $\Vert\psi\Vert(r,\sigma,w)\left\{
\begin{array}
[c]{rl}%
1 & \textnormal{ if }(r, \sigma,w) \models \psi\\
0 & \textnormal{ otherwise}%
\end{array}
,\right.  $

     \item[-]  $\Vert \ps_1 \oplus \ps_2 \Vert(r,\sigma, w) =\Vert \ps_1 \Vert(r,\sigma, w) + \Vert  \ps_2 \Vert(r,\sigma, w)$, 
 
\item[-]  $\Vert \ps_1 \otimes \ps_2 \Vert(r,\sigma, w) =\Vert \ps_1 \Vert(r,\sigma, w) \cdot \Vert  \ps_2 \Vert(r,\sigma, w)$,

 \item[-]  $\Vert \ps_1 \odot \ps_2 \Vert(r,\sigma, w) =\sum\nolimits_{w=w_1w_2}(\Vert\ps_1\Vert(r,\sigma,w_1) \cdot \Vert\ps_2\Vert(r,\sigma, w_2)) $,

\item[-]  $\Vert \ps_1 \varpi \ps_2 \Vert(r,\sigma, w) =\sum\nolimits_{w\in w_1\shuffle w_2}(\Vert\ps_1\Vert(r,\sigma,w_1) \cdot \Vert\ps_2\Vert(r, \sigma,w_2))$,

     \item[-]   $\left\Vert {\textstyle\sum x^{(i)}}. \ps \right \Vert(r,\sigma, w) = \sum\limits_{j \in [r(i)]}  \Vert \ps   \Vert (r, \sigma[x^{(i)}\rightarrow j],w)$,

     \item[-]   $\left\Vert {\textstyle\prod x^{(i)}}. \ps \right \Vert(r,\sigma, w) = \prod\limits_{j \in [r(i)]}  \Vert \ps   \Vert (r, \sigma[x^{(i)}\rightarrow j],w)$,

 \item[-]   $\left\Vert{\textstyle\sum\nolimits{^{\odot}}} x^{(i)}. \ps\right\Vert(r,\sigma,w) = \sum\limits_{1 \leq l_1 < \ldots <l_t \leq r(i)}\sum\limits_{w=w_{l_1}\ldots w_{l_t}}\prod\limits_{j=l_1, \ldots,l_t}  \Vert \ps   \Vert (r, \sigma[x^{(i)}\rightarrow j],w_j)$,

 \item[-]   $\left\Vert{\textstyle\prod\nolimits{^{\odot}}x^{(i)}}. \ps\right\Vert(r,\sigma,w) = \sum\limits_{w=w_1\ldots w_{r(i)}}\prod\limits_{1\leq j \leq r(i)}  \Vert \ps   \Vert (r, \sigma[x^{(i)}\rightarrow j],w_j)$,

\item[-]   $\left\Vert{\textstyle\sum\nolimits{^{\varpi}}} x^{(i)}. \ps\right\Vert(r,\sigma,w) = \sum\limits_{1 \leq l_1 < \ldots <l_t \leq r(i)}\sum\limits_{w\in w_{l_1} \shuffle \ldots \shuffle w_{l_t}}\prod\limits_{j=l_1, \ldots,l_t}  \Vert \ps   \Vert (r, \sigma[x^{(i)}\rightarrow j],w_j)$,

\item[-]   $\left\Vert{\textstyle\prod\nolimits{^{\varpi}}} x^{(i)}. \ps\right\Vert(r,\sigma,w) = \sum\limits_{w\in w_1 \shuffle \ldots \shuffle w_{r(i)}}\prod\limits_{1 \leq j \leq r(i)}  \Vert \ps   \Vert (r, \sigma[x^{(i)}\rightarrow j],w_j)$.    
\end{itemize}
\end{definition}

If $\ps$ is a wFOEIL sentence over $pw\B$ and $K$, then we simply write $\Vert\ps\Vert(r, w)$. Let also $\ps'$ be a wFOEIL sentence over $pw\mathcal{B}$ and $K$. Then,  $\ps$ and $\ps'$ are called \emph{equivalent w.r.t.} $r$ whenever $\Vert\ps\Vert(r, w)=\Vert\ps'\Vert(r, w)$, for every $w  \in I_{p\B(r)}^*$.

Now we are ready to formally define the concept of a parametric weighted component-based system.

\begin{definition}
A \emph{parametric weighted component-based system over} $K$ is a pair $(pw\B, \ps)$ where $pw\B=\{wB(i,j) \mid i\in [n], j \geq 1 \}$ is a set of parametric weighted components and $\ps$ is a \emph{wFOEIL} sentence over $pw\B$ and $K$.  
\end{definition}

In the sequel, for simplicity we refer to parametric weighted component-based systems simply as parametric weighted systems. We remind that in this work we focus on the architectures of parametric weighted systems. The study of parametric weighted systems' behavior is left for investigation in subsequent work as a part of parametric quantitative verification.

For our examples in the next subsection, we need the following macro wFOEIL formula. Let $pw\B=\{wB(i,j) \mid i\in [n], j \geq 1\}$ and $1 \leq i_1, \ldots, i_m \leq n$ be pairwise different indices. Let $p_{i_1} \in P(i_1), \ldots, p_{i_m} \in P(i_m)$ and $k_{i_1}, \ldots, k_{i_m}$ denote the weights in $K$ assigned to $p_{i_1}, \ldots, p_{i_m}$, respectively, i.e., $k_{i_1} =wt(i_1)(p_{i_1}), \ldots,  k_{i_m}=wt(i_m)(p_{i_m})$. We set
\begin{multline*}
\#_w\big(p_{i_1}(x^{(i_1)}) \otimes \ldots \otimes p_{i_m}(x^{(i_m)})\big)::= \big ((k_{i_1}\otimes p_{i_1}(x^{(i_1)})) \otimes  \ldots \otimes (k_{i_m}\otimes p_{i_m}(x^{(i_m)}))\big) \otimes \\ 
\bigg(\bigg(\bigwedge_{j =i_1, \ldots, i_m } \bigwedge_{p \in P(j)\setminus \{ p_j\}} \neg p(x^{(j)}) \bigg) \wedge 
\bigg(\bigwedge_{j=i_1, \ldots, i_m} \forall y^{(j)}(y^{(j)} \neq x^{(j)}).\bigwedge_{p \in P(j)} \neg p(y^{(j)}) \bigg ) \wedge \\
 \bigg( \bigwedge_{t \in [n]\setminus \{i_1, \ldots, i_m\} }\bigwedge_{p \in P(t)}\forall x^{(t)} .   \neg p(x^{(t)}) \bigg )\bigg) .
\end{multline*}

The weighted conjunctions in the right-hand side of the first line, in the above formula,   express that the ports appearing in the argument of $\#_w$ participate in the interaction with their corresponding weights. In the second line, the   
  double indexed conjunctions in the first pair of big parentheses disable all the other ports of the participating instances of components of type $i_1, \ldots, i_m$ described by variables $x^{(i_1)}, \ldots, x^{(i_m)}$, respectively; conjunctions in the second pair of parentheses disable all ports of remaining instances of component types $i_1, \ldots, i_m$. Finally, the last conjunct in the third line ensures that no ports in instances of remaining component types
participate in the interaction. Then we get 
\begin{multline*}
\#_w\big(p_{i_1}(x^{(i_1)}) \otimes \ldots \otimes p_{i_m}(x^{(i_m)})\big)\equiv \big (k_{i_1}\otimes \ldots \otimes k_{i_m}\big) \otimes \bigg(\big(p_{i_1}(x^{(i_1)}) \wedge  \ldots \wedge p_{i_m}(x^{(i_m)})\big) \\ \wedge 
\bigg(\bigwedge_{j =i_1, \ldots, i_m } \bigwedge_{p \in P(j)\setminus \{ p_j\}} \neg p(x^{(j)}) \bigg) \wedge 
\bigg(\bigwedge_{j=i_1, \ldots, i_m} \forall y^{(j)}(y^{(j)} \neq x^{(j)}).\bigwedge_{p \in P(j)} \neg p(y^{(j)}) \bigg ) \wedge \\
 \bigg( \bigwedge_{t \in [n]\setminus \{i_1, \ldots, i_m\} }\bigwedge_{p \in P(t)}\forall x^{(t)} .   \neg p(x^{(t)}) \bigg )\bigg) .
\end{multline*}

\subsection{Examples of wFOEIL sentences for parametric weighted architectures} \label{chap:examp}

In this subsection we present several examples of wFOEIL sentences describing concrete parametric architectures with quantitative features.

\begin{example}
\label{ma-sl}
\textbf{\emph{(Weighted Master/Slave)}} We present a \emph{wFOEIL} sentence for the parametric weighted Master/Slave architecture. Master/Slave architecture concerns two types of components, namely masters and slaves \cite{Ma:Co}. Every slave must be connected with exactly one master. Interactions among masters (resp. slaves) are not permitted.  
   
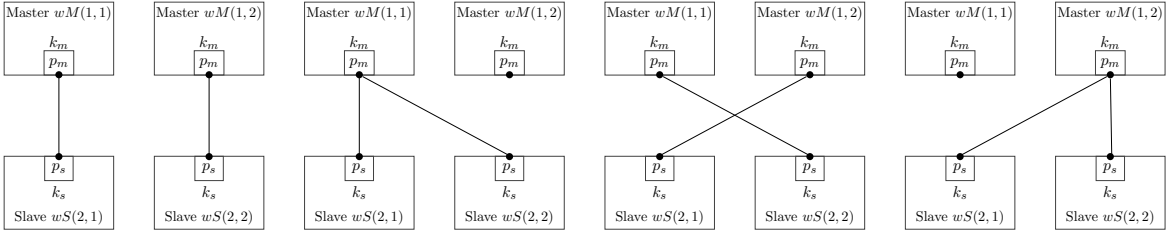
\begin{figure}[h]
\centering
\resizebox{1.0\linewidth}{!}{
\begin{tikzpicture}[>=stealth',shorten >=1pt,auto,node distance=1cm,baseline=(current bounding box.north)]
\tikzstyle{component}=[rectangle,ultra thin,draw=black!75,align=center,inner sep=9pt,minimum size=2.5cm,minimum height=1.8cm, minimum width=2.7cm]
\tikzstyle{port}=[rectangle,ultra thin,draw=black!75,minimum size=6mm]
\tikzstyle{bubble} = [fill,shape=circle,minimum size=5pt,inner sep=0pt]
\tikzstyle{type} = [draw=none,fill=none]

 \node [component,align=center] (a1)  {};
 \node [port] (a2) [below=-0.605cm of a1]  {$p_m$};
 \node[bubble] (a3) [below=-0.105cm of a1]   {};
 
\node []  (w1)  [above=0.4cm of a3]  {$k_{m}$}; 
\node[type]  [above=-0.6cm of a1]{{\small Master $wM(1,1)$}};

\node [component] (a4) [below=2cm of a1]  {};
\node [port,align=center,inner sep=5pt] (a5) [above=-0.6035cm of a4]  {$p_s$};
\node[bubble] (a6) [above=-0.105cm of a4]   {};
\node []  (w2)  [below=0.5cm of a6]  {$k_{s}$}; 
\node[type]  [below=-0.6cm of a4]{{\small Slave $wS(2,1)$}};

\path[-]          (a1)  edge                  node {} (a4);

 \node [component] (b1) [right=1cm of a1] {};
 \node [port] (b2) [below=-0.605cm of b1]  {$p_m$};
 \node[bubble] (b3) [below=-0.105cm of b1]   {};
 \node []  (w3)  [above=0.4cm of b3]  {$k_{m}$}; 
 \node[type]  [above=-0.6cm of b1]{{\small Master $wM(1,2)$}};

\node [component] (b4) [below=2cm of b1]  {};
\node [port,align=center,inner sep=5pt] (b5) [above=-0.6035cm of b4]  {$p_s$};
\node[bubble] (b6) [above=-0.105cm of b4]   {};
\node []  (w4)  [below=0.5cm of b6]  {$k_{s}$}; 
\node[type]  [below=-0.6cm of b4]{{\small Slave $wS(2,2)$}};

\path[-]          (b1)  edge                  node {} (b4);

 \node [component] (c1)[right=1cm of b1] {};
 \node [port] (c2) [below=-0.605cm of c1]  {$p_m$};
 \node[bubble] (c3) [below=-0.105cm of c1]   {};
 \node []  (w5)  [above=0.4cm of c3]  {$k_{m}$}; 
  \node[type]  [above=-0.6cm of c1]{{\small Master $wM(1,1)$}};

\node [component] (c4) [below=2cm of c1]  {};
\node [port,align=center,inner sep=5pt] (c5) [above=-0.6035cm of c4]  {$p_s$};
\node[bubble] (c6) [above=-0.105cm of c4]   {};
\node []  (w6)  [below=0.5cm of c6]  {$k_{s}$}; 
\node[type]  [below=-0.6cm of c4]{{\small Slave $wS(2,1)$}};

\path[-]          (c1)  edge                  node {} (c4);

 \node [component] (d1)[right=1cm of c1] {};
 \node [port] (d2) [below=-0.605cm of d1]  {$p_m$};
 \node[bubble] (d3) [below=-0.105cm of d1]   {};
 \node []  (w7)  [above=0.4cm of d3]  {$k_{m}$}; 
  \node[type]  [above=-0.6cm of d1]{{\small Master $wM(1,2)$}};

\node [component] (e4) [below=2cm of d1]  {};
\node [port,align=center,inner sep=5pt] (e5) [above=-0.6035cm of e4]  {$p_s$};
\node[] (i1) [above right=-0.25 cm and -0.25cm of e5]   {};
\node[bubble] (e6) [above=-0.105cm of e4]   {};
\node []  (w8)  [below=0.5cm of e6]  {$k_{s}$}; 
\node[type]  [below=-0.6cm of e4]{{\small Slave $wS(2,2)$}};

\path[-]          (c3)  edge  node {}           (i1);

 \node [component] (f1)[right=1cm of d1] {};
 \node [port] (f2) [below=-0.605cm of f1]  {$p_m$};
 \node[bubble] (f3) [below=-0.105cm of f1]   {};
 \node []  (w9)  [above=0.4cm of f3]  {$k_{m}$}; 
  \node[type]  [above=-0.6cm of f1]{{\small Master $wM(1,1)$}};

\node [component] (g4) [below=2cm of f1]  {};
\node [port,align=center,inner sep=5pt] (g5) [above=-0.6035cm of g4]  {$p_s$};
\node[] (i2) [above left=-0.25 cm and -0.25cm of g5]   {};
\node[bubble] (g6) [above=-0.105cm of g4]   {};
\node []  (w10)  [below=0.5cm of g6]  {$k_{s}$}; 
\node[type]  [below=-0.6cm of g4]{{\small Slave $wS(2,1)$}};

 \node [component] (h1)[right=1cm of f1] {};
 \node [port] (h2) [below=-0.605cm of h1]  {$p_m$};
 \node[bubble] (h3) [below=-0.105cm of h1]   {};
 \node []  (w11)  [above=0.4cm of h3]  {$k_{m}$}; 
  \node[type]  [above=-0.6cm of h1]{{\small Master $wM(1,2)$}};

\node [component] (j4) [below=2cm of h1]  {};
\node [port,align=center,inner sep=5pt] (j5) [above=-0.6035cm of j4]  {$p_s$};
\node[] (i3) [above right=-0.25 cm and -0.25cm of j5]   {};
\node[bubble] (j6) [above=-0.105cm of j4]   {};
\node []  (w12)  [below=0.5cm of j6]  {$k_{s}$}; 
\node[type]  [below=-0.6cm of j4]{{\small Slave $wS(2,2)$}};

\path[-]          (h3)  edge                  node {} (i2);

\path[-]          (f3)  edge                  node {} (i3);

 \node [component] (k1)[right=1cm of h1] {};
 \node [port] (k2) [below=-0.605cm of k1]  {$p_m$};
 \node[bubble] (k3) [below=-0.105cm of k1]   {};
 \node []  (w13)  [above=0.4cm of k3]  {$k_{m}$}; 
  \node[type]  [above=-0.6cm of k1]{{\small Master $wM(1,1)$}};

\node [component] (k4) [below=2cm of k1]  {};
\node [port,align=center,inner sep=5pt] (k5) [above=-0.6035cm of k4]  {$p_s$};
\node[] (i4) [above left=-0.25 cm and -0.25cm of k5]   {};
\node[bubble] (k6) [above=-0.105cm of k4]   {};
\node []  (w14)  [below=0.5cm of k6]  {$k_{s}$}; 
\node[type]  [below=-0.6cm of k4]{{\small Slave $wS(2,1)$}};

 \node [component] (l1)[right=1cm of k1] {};
 \node [port] (l2) [below=-0.605cm of l1]  {$p_m$};
 \node[bubble] (l3) [below=-0.105cm of l1]   {};
 \node []  (w15)  [above=0.4cm of l3]  {$k_{m}$}; 
  \node[type]  [above=-0.6cm of l1]{{\small Master $wM(1,2)$}};

\node [component] (m4) [below=2cm of l1]  {};
\node [port,align=center,inner sep=5pt] (m5) [above=-0.6035cm of m4]  {$p_s$};
\node[] (i5) [above right=-0.25 cm and -0.45cm of m5]   {};
\node[bubble] (m6) [above=-0.105cm of m4]   {};
\node []  (w16)  [below=0.5cm of m6]  {$k_{s}$}; 
\node[type]  [below=-0.6cm of m4]{{\small Slave $wS(2,2)$}};

\path[-]          (l3)  edge                  node {} (i4);

\path[-]          (l3)  edge                  node {} (i5);

\end{tikzpicture}}
\caption{Weighted Master/Slave architecture.}
\label{w-m-s}
\end{figure}

\noindent We let $\mathcal{X}^{(1)}, \mathcal{X}^{(2)}$ denote the sets of variables of master and slave weighted component instances, respectively. We denote by $p_m$ the port of master weighted component and by $p_s$ the port of slave weighted component. Then, the \emph{wFOEIL} sentence $\ps$ representing parametric weighted Master/Slave architecture is
$$\ps={\textstyle\prod\nolimits{^{\odot}}x^{(2)}} {\textstyle\sum x^{(1)}}. \#_w( p_m(x^{(1)}) \otimes p_s(x^{(2)})). $$

\noindent An instantiation of the weighted parametric Master/Slave architecture for two masters and two slaves, i.e., for  $r(1)=r(2)=2$ is shown in Figure \ref{w-m-s}. Let $\{p_m(1),p_m(2),p_s(1),p_s(2)\}$ be the set of ports. Then, consider $w_1=\{p_m(1),p_s(1)\}\{p_m(2),p_s(2)\}$, $w_2=\{p_m(1),p_s(1)\}\{p_m(1),p_s(2)\}$,
 $w_3=\{p_m(2),p_s(1)\}\{p_m(1),p_s(2)\}$, and $w_4=\{p_m(2),p_s(1)\}\{p_m(2),p_s(2)\}$, that correspond to the 
four possible connections for the components in the system as shown in Figure \ref{w-m-s}. Then, the values
$\Vert \ps \Vert (r,w_1)$, $\Vert \ps \Vert (r,w_2)$, $\Vert \ps \Vert (r,w_3)$,
 and $\Vert \ps \Vert (r,w_4)$
return the cost of the implementation of each of the four possible interactions in the architecture, according to the underlying semiring.
In turn, the ``sum" $\Vert \ps \Vert (r,w_1)+\Vert \ps \Vert (r,w_2)+\Vert \ps \Vert (r,w_3)
 +\Vert \ps \Vert (r,w_4)$
equals for instance, in the semiring of natural numbers to the total cost for executing
the possible connections in the architecture.

\end{example}

\begin{example}
\label{str}
\textbf{\emph{(Weighted Star)}} Star architecture has only one component type with a unique port, namely $p$. One instance is considered as the center in the sense that every other instance has to be connected with it. No other interaction is permitted. Figure \ref{w-star} represents the Star architecture for five instances and center $wS(1,1)$.

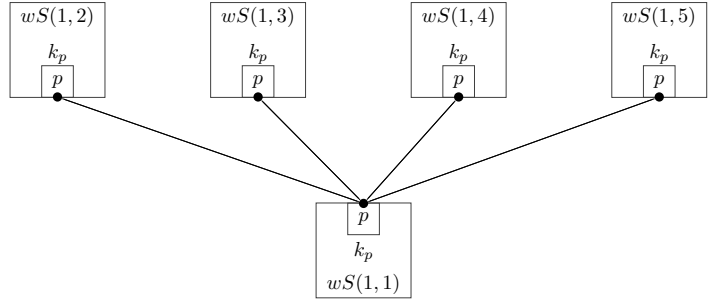
\begin{figure}[h]
\centering
\resizebox{0.6\linewidth}{!}{
\begin{tikzpicture}[>=stealth',shorten >=1pt,auto,node distance=1cm,baseline=(current bounding box.north)]
\tikzstyle{component}=[rectangle,ultra thin,draw=black!75,align=center,inner sep=9pt,minimum size=1.8cm]
\tikzstyle{port}=[rectangle,ultra thin,draw=black!75,minimum size=6mm]
\tikzstyle{bubble} = [fill,shape=circle,minimum size=5pt,inner sep=0pt]
\tikzstyle{type} = [draw=none,fill=none] 

\node [component] (a1) {};
\node [port] (a2) [above=-0.605cm of a1]  {$p$};
\node[bubble] (a3) [above=-0.105cm of a1]   {};
 
 \node []  (w1)  [below=0.5cm of a3]  {$k_{p}$};  
 \node []  (s1)  [below=1.1cm of a3]  {$wS(1,1)$}; 

\node [component] (a4) [above left =2.cm and 4cm of a1]  {};
\node [port] (a5) [below=-0.605cm of a4]  {$p$};
\node[] (i1) [above left=-0.15 cm and -0.73cm of a4]   {};
\node[bubble] (a6) [below=-0.105cm of a4]   {};

\node []  (w2)  [above=0.4cm of a6]  {$k_{p}$}; 
 \node []  (s2)  [above=1.1cm of a6]  {$wS(1,2)$}; 

\path[-]          (a3)  edge                  node {} (a6);
\path[-]          (a6)  edge                  node {} (a3);

\node [component] (b4) [right =2 cm of a4]  {};
\node [port] (b5) [below=-0.605cm of b4]  {$p$};
\node[] (i2) [above left=-0.25 cm and -0.30cm of b5]   {};
\node[bubble] (b6) [below=-0.105cm of b4]   {};
\node []  (w3)  [above=0.4cm of b6]  {$k_{p}$}; 
 \node []  (s3)  [above=1.1cm of b6]  {$wS(1,3)$};

\path[-]          (a3)  edge                  node {} (b6);
\path[-]          (b6)  edge                  node {} (a3);
 
\node [component] (c4) [right =2 cm of b4]  {};
\node [port] (c5) [below=-0.605cm of c4]  {$p$};
\node[] (i3) [above right=-0.20 cm and -0.38cm of c5]   {};
\node[bubble] (c6) [below=-0.105cm of c4]   {};
\node []  (w4)  [above=0.4cm of c6]  {$k_{p}$}; 
 \node []  (s4)  [above=1.1cm of c6]  {$wS(1,4)$}; 

\path[-]          (a3)  edge                  node {} (c6);
\path[-]          (c6)  edge                  node {} (a3);
 
\node [component] (d4) [right =2 cm of c4]  {};
\node [port] (d5) [below=-0.605cm of d4]  {$p$};
\node[] (i4) [above right=-0.20 cm and -0.27cm of d5]   {};
\node[bubble] (d6) [below=-0.105cm of d4]   {};
\node []  (w5)  [above=0.4cm of d6]  {$k_{p}$}; 
 \node []  (s5)  [above=1.1cm of d6]  {$wS(1,5)$}; 
\path[-]          (a3)  edge                  node {} (d6);
\path[-]          (d6)  edge                  node {} (a3);

\end{tikzpicture}}
\caption{Weighted Star architecture.}
\label{w-star}
\end{figure}

\noindent The \emph{wFOEIL} sentence $\ps$ for the parametric weighted Star architecture is as follows:
$$ \ps= {\textstyle\sum x^{(1)}}  {\textstyle\prod\nolimits{^{\odot}}y^{(1)}} (x^{(1)}\neq y^{(1)}). \#_w(p(x^{(1)}) \otimes p(y^{(1)})). $$

\noindent Let $a_1=\lbrace p(1),p(2)\rbrace,a_2=\lbrace p(1),p(3)\rbrace,a_3=\lbrace p(1),p(4)\rbrace,a_4=\lbrace p(1),p(5)\rbrace$ refer to the
interactions of $wS(1,2)$, $wS(1,3)$, $wS(1,4)$, and $wS(1,5)$ respectively, to $wS(1,1)$. 
The value $\Vert \ps \Vert (r,w)$ for $w=a_1a_2a_3a_4$ is the cost of the
 implementation of this architecture. Similarly, we get the cost of all possible 
Star architectures with center $wS(1,2),wS(1,3), wS(1,4), wS(1,5)$, respectively. Then, if we ``sum up" 
those values we get for instance, in the semiring of rational numbers, the total cost. 
\end{example}

\begin{example}
\label{pi-fi}
\textbf{\emph{(Weighted Pipes/Filters)}} Pipes/Filters architecture involves two types of components, namely pipes and filters \cite{Ga:An}. Pipe (resp. filter) component has an entry  port  $p_e$ and an output port $p_o$ (resp. $f_e,f_o$). Every filter $F$ is connected to two separate pipes $P$ and $P'$ via interactions  $\{f_e, p_o\}$ and $\{f_o,p'_e \}$, respectively. Every pipe $P$ can be connected to at most one filter $F$ via an interaction $\{p_o, f_e\}$. Any other interaction is not permitted. Figure \ref{w-p-f} shows an instantiation of the parametric weighted Pipes/Filters architecture for four pipes and three filters,
i.e., for $r(1)=4$ and $r(2)=3$. 

\begin{figure}[h]
\centering
\resizebox{1.0\linewidth}{!}{
\begin{tikzpicture}[>=stealth',shorten >=1pt,auto,node distance=1cm,baseline=(current bounding box.north)]
\tikzstyle{component}=[rectangle,ultra thin,draw=black!75,align=center,inner sep=9pt,minimum size=1.5cm,minimum width=2.5cm,minimum height=2cm]
\tikzstyle{port}=[rectangle,ultra thin,draw=black!75,minimum size=6mm]
\tikzstyle{bubble} = [fill,shape=circle,minimum size=5pt,inner sep=0pt]
\tikzstyle{type} = [draw=none,fill=none]

 \node [component] (a1) {};

\node[bubble] (a2) [right=-0.105cm of a1]   {};   
\node [port] (a3) [right=-0.605cm of a1]  {$p_e$};  
\node []  (w1)  [below=0.0cm of a3]  {$k_{p_e}$};

\node[bubble] (a4) [left=-0.105cm of a1]   {};   
\node [port] (a5) [left=-0.615cm of a1]  {$p_o$};  
\node []  (w2)  [below=0.0cm of a5]  {$k_{p_o}$};
\node[type]  [above=-0.6cm of a1]{{\small Pipe $P(2,1)$}};

\node [component] (b2) [right=2cm of a1]  {};
\node[bubble] (b3) [right=-0.105cm of b2]   {};   
\node [port] (b4) [right=-0.605cm of b2]  {$f_e$};  
\node []  (w3)  [below=0.0cm of b4]  {$k_{f_e}$};

\node[bubble] (b5) [left=-0.105cm of b2]   {};   
\node [port] (b6) [left=-0.615cm of b2]  {$f_o$};  
\node []  (w4)  [below=0.0cm of b6]  {$k_{f_o}$};
\node[type]  [above=-0.6cm of b2]{{\small Filter $F(1,1)$}};
\path[-]          (a3)  edge                  node {} (b6);

 \node [component] (c2) [right=2cm of b2]{};

\node[bubble] (c3) [right=-0.105cm of c2]   {};   
\node [port] (c4) [right=-0.61cm of c2]  {$p_e$};  
\node []  (w5)  [below=0.0cm of c4]  {$k_{p_e}$};

\node[bubble] (c5) [left=-0.105cm of c2]   {};   
\node [port] (c6) [left=-0.615cm of c2]  {$p_o$};  
\node []  (w6)  [below=0.0cm of c6]  {$k_{p_o}$};
\node[type]  [above=-0.6cm of c2]{{\small Pipe $P(2,2)$}};

\path[-]          (b4)  edge                  node {} (c6);

\node [component] (d2) [above right= 0.5cm and 2cm of c2]  {};
\node[bubble] (d3) [right=-0.105cm of d2]   {};  
\node [port] (d4) [right=-0.600cm of d2]  {$f_e$};  
\node []  (w7)  [below=0.0cm of d4]  {$k_{f_e}$};

\node[bubble] (d5) [left=-0.105cm of d2]   {};   
\node [port] (d6) [left=-0.61cm of d2]  {$f_o$};  
\node[] (i1) [above left=-0.3 cm and -0.20cm of d6]   {};
\node []  (w8)  [below=0.0cm of d6]  {$k_{f_o}$};
\node[type]  [above=-0.6cm of d2]{Filter $F(1,2)$};
\path[-]          (c3)  edge                  node {} (i1);

\node [component] (e2) [below right= 0.5cm and 2cm of c2]  {};
\node[bubble] (e3) [right=-0.105cm of e2]   {};   
\node [port] (e4) [right=-0.600cm of e2]  {$f_e$};  
\node []  (w9)  [below=0.0cm of e4]  {$k_{f_e}$};

\node[bubble] (e5) [left=-0.105cm of e2]   {};   
\node [port] (e6) [left=-0.61cm of e2]  {$f_o$}; 
\node[] (i2) [above left=-0.62 cm and -0.30cm of e6]   {};
\node []  (w10)  [below=0.0cm of e6]  {$k_{f_o}$};

\node[type]  [above=-0.6cm of e2]{{\small Filter $F(1,3)$}};
\path[-]          (c3)  edge                  node {} (i2);

 \node [component] (f2) [right=2cm of d2]{};

\node[bubble] (f3) [right=-0.105cm of f2]   {};   
\node [port] (f4) [right=-0.61cm of f2]  {$p_e$}; 
\node []  (w11)  [below=0.0cm of f4]  {$k_{p_e}$};

\node[bubble] (f5) [left=-0.105cm of f2]   {};   
\node [port] (f6) [left=-0.615cm of f2]  {$p_o$};  
\node[] (i3) [above left=-0.43 cm and -0.25cm of f6]   {};
\node []  (w12)  [below=0.0cm of f6]  {$k_{p_o}$};

\node[type]  [above=-0.6cm of f2]{{\small Pipe $P(2,3)$}};

\path[-]          (d3)  edge                  node {} (i3);

 \node [component] (g2) [right=2cm of e2]{};

\node[bubble] (g3) [right=-0.105cm of g2]   {};   
\node [port] (g4) [right=-0.61cm of g2]  {$p_e$};  
\node []  (w13)  [below=0.0cm of g4]  {$k_{p_e}$};

\node[bubble] (g5) [left=-0.105cm of g2]   {};   
\node [port] (g6) [left=-0.615cm of g2]  {$p_o$};  
\node[] (i4) [above left=-0.43 cm and -0.25cm of g6]   {};
\node []  (w14)  [below=0.0cm of g6]  {$k_{p_o}$};

\node[type]  [above=-0.6cm of g2]{{\small Pipe $P(2,4)$}};

\path[-]          (e3)  edge                  node {} (i4);

\end{tikzpicture}}
\caption{Weighted Pipes/Filters architecture}
\label{w-p-f}
\end{figure}
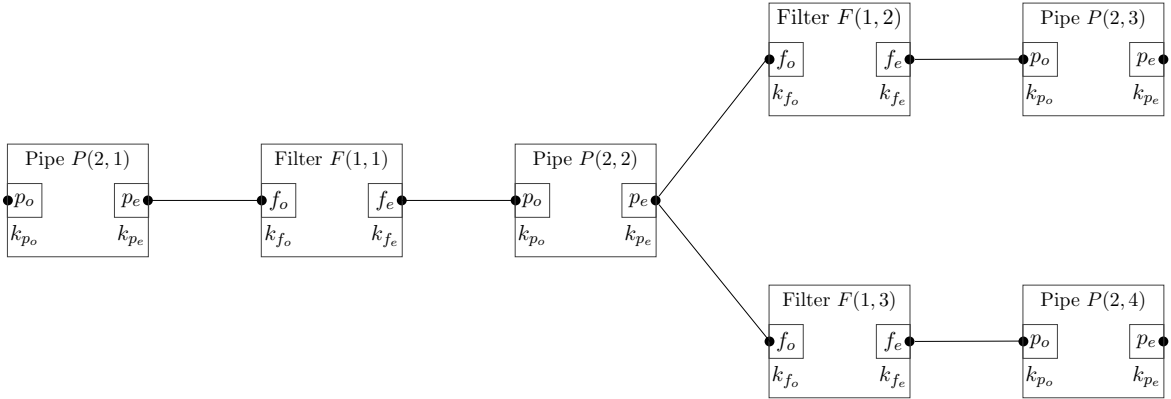

\noindent We denote by $\mathcal{X}^{(1)}$ and $\mathcal{X}^{(2)}$ the sets of variables  corresponding to pipe and filter weighted component instances, respectively. The subsequent \emph{wFOEIL} sentence $\ps$ describes the parametric weighted Pipes/Filters architecture.
\begin{multline*}
\ps= {\textstyle\prod\nolimits{^{\odot}} x^{(2)}}{\textstyle\sum x^{(1)}} {\textstyle\sum y^{(1)}}(x^{(1)} \neq y^{(1)}).\bigg(\#_w(p_o(x^{(1)}) \otimes f_e(x^{(2)}))\odot    \#_w(p_e(y^{(1)}) \otimes f_o(x^{(2)})) \bigg) \bigotimes     \\
\qquad \quad \bigg ( \forall z^{(1)}\forall y^{(2)}  \bigg(\forall z^{(2)} (y^{(2)} \neq z^{(2)}).  
\big( \big(\mathrm{true}* (p_o(z^{(1)}) \wedge f_e(y^{(2)}))*\mathrm{true}\big) \wedge \\ \qquad \qquad \qquad \big(\neg \big( \mathrm{true}* (p_o(z^{(1)}) \wedge f_e(z^{(2)}))*
\mathrm{true}\big)\big)\big)\bigg) \vee \hide{\\
\qquad \qquad \qquad \qquad \qquad \qquad} \bigg(\neg\big( \mathrm{true} *  (p_o(z^{(1)}) \wedge f_e(y^{(2)}))  * 
\mathrm{true}\big)\bigg) \bigg).
\end{multline*}
\noindent In the above weighted sentence the arguments of $\#_w$ express the cost for the connection of a filter entry (resp. output) port with a pipe output (resp. entry) port. The \emph{FOEIL} subformula after the big $\otimes$ ensures that no more than one filter entry port will be connected to the same pipe output port.

\noindent Let $w_1=\{f_e(1),p_o(2)\}\{f_o(1),p_e(1)\}\{f_e(2),p_o(3)\}\{f_o(2),p_e(2)\}\{f_e(3),p_o(4)\}\{f_o(3),p_e(2)\}$ and $w_2=\{f_e(1),p_o(3)\}\{f_o(1),p_e(4)\}\{f_e(2),p_o(4)\}\{f_o(2),p_e(1)\}\{f_e(3),p_o(2)\}\{f_o(3),p_e(4)\}$ that encode 
two possible implementations of the interactions for the architecture instantiation of Figure \ref{w-p-f}.
 Then, the value $\Vert \ps \Vert (r,w_1)+\Vert \ps \Vert (r,w_2)$ represents for instance, in the semiring of natural numbers the total cost of performing the interactions.

\end{example}

\begin{example}
\label{repo}
\textbf{\emph{(Weighted Repository)}} Repository architecture involves two types of components, namely repository and data accessor \cite{Cl:Do}. Repository component is unique and all data accessors are connected to it. No connection among data accessors exists. Both repository and data accessors have one port each called $p_r,p_d$, respectively. Figure \ref{w-rep} shows the instantiation of parametric weighted Repository architecture for four data accessors. 

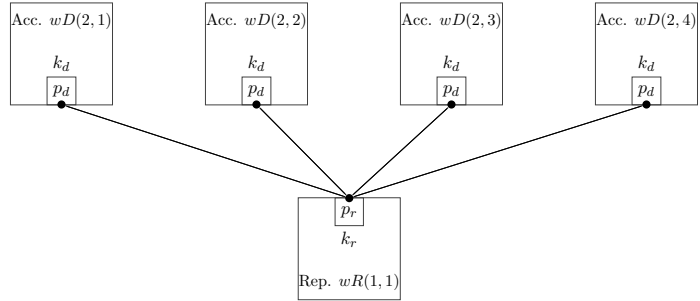
\begin{figure}[h]
\centering
\resizebox{0.6\linewidth}{!}{
\begin{tikzpicture}[>=stealth',shorten >=1pt,auto,node distance=1cm,baseline=(current bounding box.north)]
\tikzstyle{component}=[rectangle,ultra thin,draw=black!75,align=center,inner sep=9pt,minimum size=2.2cm]
\tikzstyle{port}=[rectangle,ultra thin,draw=black!75,minimum size=6mm]
\tikzstyle{bubble} = [fill,shape=circle,minimum size=5pt,inner sep=0pt]
\tikzstyle{type} = [draw=none,fill=none] 

\node [component] (a1) {};
\node [port] (a2) [above=-0.605cm of a1]  {$p_r$};
\node[bubble] (a3) [above=-0.105cm of a1]   {};
 
 \node []  (w1)  [below=0.5cm of a3]  {$k_{r}$};  
 \node []  (s1)  [below=1.4cm of a3]  {{\small Rep. $wR(1,1)$}};

\node [component] (a4) [above left =2.0cm and 4cm of a1]  {};
\node [port] (a5) [below=-0.605cm of a4]  {$p_d$};
\node[] (i1) [above left=-0.15 cm and -0.73cm of a4]   {};
\node[bubble] (a6) [below=-0.105cm of a4]   {};

\node []  (w2)  [above=0.5cm of a6]  {$k_{d}$}; 
\node []  (s2)  [above=1.4cm of a6]  {{\small Acc. $wD(2,1)$}};

\path[-]          (a3)  edge                  node {} (a6);
\path[-]          (a6)  edge                  node {} (a3);

\node [component] (b4) [right =2 cm of a4]  {};
\node [port] (b5) [below=-0.605cm of b4]  {$p_d$};
\node[] (i2) [above left=-0.25 cm and -0.30cm of b5]   {};
\node[bubble] (b6) [below=-0.105cm of b4]   {};
\node []  (w3)  [above=0.5cm of b6]  {$k_{d}$}; 
\node []  (s3)  [above=1.4cm of b6]  {{\small Acc. $wD(2,2)$}};

\path[-]          (a3)  edge                  node {} (b6);
\path[-]          (b6)  edge                  node {} (a3);
 
\node [component] (c4) [right =2 cm of b4]  {};
\node [port] (c5) [below=-0.605cm of c4]  {$p_d$};
\node[] (i3) [above right=-0.20 cm and -0.38cm of c5]   {};
\node[bubble] (c6) [below=-0.105cm of c4]   {};
\node []  (w4)  [above=0.5cm of c6]  {$k_{d}$};
\node []  (s4)  [above=1.4cm of c6]  {{\small Acc. $wD(2,3)$}};
 
\path[-]          (a3)  edge                  node {} (c6);
\path[-]          (c6)  edge                  node {} (a3);
 
\node [component] (d4) [right =2 cm of c4]  {};
\node [port] (d5) [below=-0.605cm of d4]  {$p_d$};
\node[] (i4) [above right=-0.20 cm and -0.27cm of d5]   {};
\node[bubble] (d6) [below=-0.105cm of d4]   {};
\node []  (w5)  [above=0.5cm of d6]  {$k_{d}$}; 
\node []  (s5)  [above=1.4cm of d6]  {{\small Acc. $wD(2,4)$}};

\path[-]          (a3)  edge                  node {} (d6);
\path[-]          (d6)  edge                  node {} (a3);

\end{tikzpicture}}
\caption{Weighted Repository architecture.}
\label{w-rep}
\end{figure}

\noindent The subsequent \emph{wFOEIL} sentence $\ps$ characterizes the parametric Repository architecture with weighted features. Variable set $\mathcal{X}^{(1)}$ refers to  instances of repository component and variable set $\mathcal{X}^{(2)}$ to instances of data accessor component.  
$$ \ps= {\textstyle\sum x^{(1)}} {\textstyle\prod\nolimits{^{\odot}} x^{(2)}}.  \#_w(p_r(x^{(1)}) \otimes
p_d(x^{(2)})). $$

\noindent Let $K$ be the fuzzy semiring, and consider $a_1= \{p_r(1),p_d(1)\}$, $a_2=\{p_r(1),p_d(2)\}$, $a_3= \{p_r(1),p_d(3)\}$, $a_4=\{p_r(1),p_d(4)\}$
 which represent each of the four connections for the architecture instantiation of Figure \ref{w-rep}. Then, the value
 $\Vert \ps \Vert (r,w)$ for $w=a_1a_2a_3a_4$ is the cost for implementing the interactions
computed in a fuzzy framework system.

\end{example}

The interactions of  parametric weighted architectures discussed in  Examples \ref{ma-sl}-\ref{repo} can be executed with arbitrary order. Hence, wFOEIL can describe sufficiently parametric weighted architectures with no order restrictions in the allowed interactions. Next, we provide three more examples of architectures with quantitative characteristics, namely weighted Blackboard, weighted Request/Response, and weighted Publish/Subscribe where the order of interactions constitutes a main feature.

\begin{example}
\label{blackboard}
\textbf{\emph{(Weighted Blackboard)}} The subsequent \emph{wFOEIL} sentence $\ps$ encodes the cost of the interactions of weighted Blackboard architecture, described in Example \ref{ex_b_blackboard}, in the parametric setting. We consider three set of variables, namely $\mathcal{X}^{(1)}, \mathcal{X}^{(2)}, \mathcal{X}^{(3)},$ for the weighted component instances of blackboard, controller, and knowledge sources components, respectively.  
\begin{multline*}
\ps={\textstyle\sum x^{(1)}} {\textstyle\sum x^{(2)}}.\bigg( \#_w (p_d(x^{(1)})\otimes p_r(x^{(2)})) \odot \bigg({\textstyle\prod^{\varpi} x^{(3)}}.\#_w (p_d(x^{(1)})\otimes p_n(x^{(3)}))\bigg) \odot \\ \bigg({\textstyle\sum^{\varpi} y^{(3)}}. \big(\#_w(p_l(x^{(2)})\otimes p_t(y^{(3)})) \odot  \#_w (p_e(x^{(2)})\otimes p_w(y^{(3)})\otimes p_a(x^{(1)}))\big) \bigg)\bigg).
\end{multline*}
\end{example}

\begin{example}
\label{re-re}
\textbf{\emph{(Weighted Request/Response)}} Next we present a \emph{wFOEIL} sentence $\ps$ for weighted Request/Response architecture, described in Example \ref{b-r-r}, in the parametric setting. We consider the variable sets $\mathcal{X}^{(1)},\mathcal{X}^{(2)},\mathcal{X}^{(3)}$, and $\mathcal{X}^{(4)}$  referring to weighted component instances of service registry, service, client, and coordinator component, respectively. 
\begin{multline*}
\ps= \bigg( {\textstyle\sum x^{(1)}}.\bigg(\big({\textstyle\prod^{\varpi} x^{(2)}}. \#_w (p_e(x^{(1)})\otimes p_r(x^{(2)}))\big) \odot \\ \big({\textstyle\prod^{\varpi} x^{(3)}}. (\#_w (p_l(x^{(3)})\otimes p_u(x^{(1)})) \odot \#_w(p_o(x^{(3)})\otimes p_t(x^{(1)}))) \big )\bigg)\bigg) \odot \\ \bigg({\textstyle\sum^{\varpi} y^{(2)}}{\textstyle\sum x^{(4)}}{\textstyle\sum^{\odot}y^{(3)}}.\x \otimes \bigg(\forall y^{(4)}\forall z^{(3)}\forall z^{(2)}.\big( \theta \vee \big(\forall t^{(3)}\forall t^{(2)}(z^{(2)}\neq t^{(2)}).\theta'\big)\big)\bigg)\bigg)
\end{multline*}
\noindent where the \emph{wEPIL} formula $\x$ is  given by:

\

$ \x= \#_w (p_n(y^{(3)})\otimes p_m(x^{(4)})) \odot \#_w (p_q(y^{(3)})\otimes p_a(x^{(4)})\otimes p_g (y^{(2)})) \odot  \\ \text{\qquad \qquad \qquad \qquad \qquad \qquad \qquad \qquad \qquad\qquad \qquad } \#_w (p_c(y^{(3)})  \otimes p_d(x^{(4)}) \otimes p_s(y^{(2)}))$,

\

\noindent and the \emph{EPIL} formulas $\theta$ and $\theta'$ are given respectively, by

$\theta=\neg (\mathrm{true}* \# (p_q(z^{(3)})\wedge p_a(y^{(4)})\wedge p_g (z^{(2)}))*\mathrm{true})$,

\

\noindent and

\

$\theta' =(\mathrm{true} *  \# (p_q(z^{(3)})\wedge p_a(y^{(4)})\wedge p_g (z^{(2)}))*\mathrm{true}) \wedge \\
 \text{ \qquad \qquad \qquad \qquad \qquad \qquad \qquad \ \ \ \ } \neg( \mathrm{true} * \# (p_q(t^{(3)})\wedge p_a(y^{(4)})\wedge p_g (t^{(2)}))* \mathrm{true}).$

\

\noindent The \emph{EPIL} subformula $\forall y^{(4)}\forall z^{(3)} \forall z^{(2)} .\big( \theta \vee \big(\forall t^{(3)}\forall t^{(2)}(z^{(2)}\neq t^{(2)}).\theta'\big)\big)$   in $\ps$ serves as a constraint to ensure that a unique coordinator is assigned to each service.    
\end{example}

\begin{example}\textbf{\emph{(Weighted Publish/Subscribe)}}
\label{pu-su}
We consider weighted Publish/Subscribe architecture, described in Example \ref{b-pu-su}, in the parametric setting. In the subsequent \emph{wFOEIL} sentence $\ps$, we let variable sets $\mathcal{X}^{(1)},\mathcal{X}^{(2)},\mathcal{X}^{(3)}$ correspond to publisher, topic, and subscriber weighted component instances, respectively. 
\begin{multline*}
\ps= {\textstyle\sum^{\varpi} x^{(2)}}. \Bigg( \bigg( {\textstyle\sum^{\varpi} x^{(1)}}.\big(\#_w (p_a(x^{(1)})\otimes p_n (x^{(2)})) \odot \#_w (p_t(x^{(1)})\otimes p_r(x^{(2)}))\big) \bigg)  
\odot \\ \bigg({\textstyle\sum^{\varpi} x^{(3)}}. \big( \#_w (p_e(x^{(3)})\otimes p_c(x^{(2)})) \odot \#_w(p_g(x^{(3)})\otimes p_s (x^{(2)})) \odot  \#_w(p_d(x^{(3)})\otimes p_f (x^{(2)}))\big)\bigg)\Bigg).
\end{multline*}
\end{example}

Existing work \cite{Bo:Ch,Bo:St,Ko:Pa,Ma:Co} studied the architectures of the presented examples in the qualitative setting. On the other hand, the work of \cite{Ka:We,Pa:On,Pa:We} consider no execution order of the weighted interactions, and the parametric setting is considered only in \cite{Pa:On}.

Observe that in the presented examples,
 whenever it is defined a unique instance for a weighted component type we may also consider the corresponding set of variables as a singleton. 
\hide{ 
\subsection{Parametric BIP ab-models}
We introduce parametric BIP ab-models consisting of sets of parametric components and FOEIL sentences over them. We define the semantics of a parametric BIP ab-model as a BIP ab-model w.r.t. a mapping which assigns a concrete number of instances to every component type. In the subsequent definition, we assume that the FOEIL sentence $\psi$ is well-defined, i.e., no more that one port in each component instance participates in every interaction. If it is not the case, we can trivially achieve this, by taking the conjunction of $\psi$ with the FOEIL sentence $\forall x \bar\forall y. \bigwedge\limits_{p \neq p'}(\neg p(x,y) \vee \neg p'(x,y)) $.

\begin{definition}
A \emph{parametric BIP ab-model} is a pair $(p\B, \psi)$ where $p\B=\{B(i,j) \mid i \in [n],  j \geq 1 \}$ is a set of parametric components and $\psi$ is a \emph{FOEIL} sentence over $p\B$.
\end{definition}

\noindent For the semantics of parametric BIP ab-models we need the next result.

\begin{proposition}\label{FOEIL_to_EPIL}
Let $p\mathcal{B}= \{B(i,j) \mid i \in [n], j \geq 1 \} $ be a set of parametric components, $\psi$ a \emph{FOEIL} formula over $p\B$, $r:[n] \rightarrow \mathbb{N}$ a mapping, \blue{$\mathcal{V}\subseteq \mathcal{X} \cup \mathcal{Y}$}  a finite set containing $\mathrm{free}(\psi)$, and $\sigma$ a $(\mathcal{V},r)$-assignment. Then, we can effectively construct an \emph{EPIL} formula $\phi_{\psi,\sigma}$ over $P_{p\B(r)}$ such that $(r,\sigma, w) \models \psi$ iff $w \models \phi_{\psi, \sigma}$ for every $w \in I(P_{p\B(r)})^*$.
\end{proposition}

\begin{proof} We construct $\phi_{\psi, \sigma}$ by induction on the structure of $\psi$ as follows.
\begin{itemize}
\item[-] If $\psi=\phi$, then $\phi_{\psi,\sigma} = \phi$.
\item[-] If $\psi=(x=x')$, then $\phi_{\psi,\sigma}=
 \left\{
\begin{array}
[c]{ll}%
\mathrm{true} & \textnormal{ if } \sigma(x)  = \sigma(x')\\
\mathrm{false} & \textnormal{ otherwise}%
\end{array}
.\right.$ 
\item[-] If $\psi=(y=y')$, then $\phi_{\psi,\sigma}=
 \left\{
\begin{array}
[c]{ll}%
\mathrm{true} & \textnormal{ if } \sigma(y)  = \sigma(y')\\
\mathrm{false} & \textnormal{ otherwise}%
\end{array}
.\right.$ 
\item[-] If $\psi=\neg(x=x')$, then $\phi_{\psi,\sigma}=
 \left\{
\begin{array}
[c]{ll}%
\mathrm{true} & \textnormal{ if } \sigma(x)  \neq \sigma(x')\\
\mathrm{false} & \textnormal{ otherwise}%
\end{array}
.\right.$ 
\item[-] If $\psi=\neg (y=y')$, then $\phi_{\psi,\sigma}=
 \left\{
\begin{array}
[c]{ll}%
\mathrm{true} & \textnormal{ if } \sigma(y)  \neq \sigma(y')\\
\mathrm{false} & \textnormal{ otherwise}%
\end{array}
.\right.$ 
\item[-] If $\psi=\psi_1 \vee \psi_2$, then $\phi_{\psi,\sigma}=\phi_{\psi_1,\sigma}\vee \phi_{\psi_2,\sigma}$.
\item[-] If $\psi=\psi_1 \wedge \psi_2$, then $\phi_{\psi,\sigma}=\phi_{\psi_1,\sigma}\wedge \phi_{\psi_2,\sigma}$.
\item[-] If $\psi = \psi_1 \barwedge \psi_2$, then $\phi_{\psi,\sigma}= \phi_{\psi_1,\sigma} \barwedge \phi_{\psi_2,\sigma}$. 
\item[-] If $\psi=\exists x\exists y . \psi'$, then $\phi_{\psi,\sigma}=\bigvee\limits_{\substack{i\in [n]}} \bigvee\limits_{\substack{j\in [r(i)]}}\phi_{\psi',\sigma[x \rightarrow i, y \rightarrow j ]}$.
\item[-] If $\psi=\exists x\forall y . \psi'$, then $\phi_{\psi,\sigma}=\bigvee\limits_{\substack{i\in [n]}} \bigwedge\limits_{\substack{j\in [r(i)]}}\phi_{\psi',\sigma[x \rightarrow i, y \rightarrow j ]}$.
\item[-] If $\psi=\exists x \bar{\forall} y . \psi'$, then $\phi_{\psi,\sigma}=\bigvee\limits_{\substack{i\in [n]}} \bar{\bigwedge\limits_{\substack{j\in [r(i)]}}}\phi_{\psi',\sigma[x \rightarrow i, y \rightarrow j ]}$.
\item[-] If $\psi=\forall x\exists y . \psi'$, then $\phi_{\psi,\sigma}=\bigwedge\limits_{\substack{i\in [n]}} \bigvee\limits_{\substack{j\in [r(i)]}}\phi_{\psi',\sigma[x \rightarrow i, y \rightarrow j ]}$.
\item[-] If $\psi=\forall x\forall y . \psi'$, then $\phi_{\psi,\sigma}=\bigwedge\limits_{\substack{i\in [n]}} \bigwedge\limits_{\substack{j\in [r(i)]}}\phi_{\psi',\sigma[x \rightarrow i, y \rightarrow j ]}$.
\item[-] If $\psi=\forall x\bar{\forall} y . \psi'$, then $\phi_{\psi,\sigma}=\bigwedge\limits_{\substack{i\in [n]}} \bar{\bigwedge\limits_{\substack{j\in [r(i)]}}}\phi_{\psi',\sigma[x \rightarrow i, y \rightarrow j ]}$.
\end{itemize}
\end{proof}

\noindent If $\psi$ is a sentence in Proposition \ref{FOEIL_to_EPIL}, then  $\mathrm{free}(\psi)= \emptyset$ hence we need no assignment.

\begin{definition}
Let $(p\B, \psi)$ be a parametric \emph{BIP} ab-model with $p\B=\{B(i,j) \mid i \in [n], j \geq 1\}$, and $r:[n] \rightarrow \mathbb{N}$. Then, the semantics of $(p\B,\psi)$ w.r.t. $r$ is the \emph{BIP} ab-model $(p\B(r),\phi_{\psi})$.
\end{definition}
} 
 
\section{Decidability results for wFOEIL}\label{sec_dec}
In this section, we state an effective translation of wFOEIL sentences to weighted automata. For this, we use a corresponding result from \cite{Pi:Ar}, namely for every FOEIL sentence we can effectively construct, in exponential time, an expressively equivalent finite automaton. Then, we show that the equivalence of wFOEIL sentences over specific semirings is decidable. For the reader's convenience we briefly recall basic notions and results on weighted automata. 

Let $A$ be an alphabet. A (nondeterministic) weighted finite automaton (WFA for short) over $A$ and $K$ is a quadruple $\mathcal{A}=(Q, in, wt, ter)$ where $Q$ is the finite state set, $in: Q \rightarrow K$ is the initial distribution, $wt:Q\times A \times Q \rightarrow K$ is the mapping assigning weights to transitions of $\mathcal{A}$, and $ter:Q \rightarrow K$ is the terminal distribution.  

Let $w=a_1 \ldots a_n \in A^*$. A path $P_w$ of $\mathcal{A}$ over $w$ is a sequence of transitions $P_w=((q_{i-1}, a_i, q_i))_{1 \leq i \leq n}$. The weight of $P_w$ is given by 
$weight(P_w)=in(q_0)\cdot\prod_{1 \leq i \leq n}wt(q_{i-1},a_i,q_i) \cdot ter(q_n)$. The behavior of $\mathcal{A}$ is the series $\Vert \mathcal{A} \Vert:A^* \rightarrow K$ which is determined by $\Vert \mathcal{A} \Vert(w)= \sum_{P_w}weight(P_w)$. 

Two WFA $\mathcal{A}$ and $\mathcal{A}'$ over $A$ and $K$ are called equivalent if $\Vert\mathcal{A}\Vert=\Vert\mathcal{A}'\Vert$.  For our translation algorithm, of wFOEIL formulas to WFA, we shall need folklore results in WFA theory. We collect them in the following proposition (cf. for instance \cite{Dr:Ha, Sa:El}). 

\begin{proposition}\label{prop-rec} Let $\mathcal{A}_1=(Q_1, in_1, wt_1, ter_1)$ and $\mathcal{A}_2=(Q_2, in_1, wt_2, ter_2)$ be two \emph{WFA's} over $A$ and $K$. Then, we can construct in polynomial time \emph{WFA's} $\mathcal{B}, \mathcal{C}, \mathcal{D}, \mathcal{E}$ over $A$ and $K$ accepting the sum, and the Hadamard, Cauchy and shuffle product of $\Vert \mathcal{A}_1\Vert$ and  $\Vert \mathcal{A}_2\Vert$, respectively. \end{proposition}

Next we present the translation algorithm of  wFOEIL formulas to WFA's. Our algorithm requires a doubly exponential time  at its worst case.  Specifically, we prove the following theorem.

\begin{theorem}\label{wsent_to_waut}
Let $pw\B=\{wB(i,j) \mid i\in [n], j\geq 1\}$ be a set of parametric weighted components over a commutative semiring $K$,  and $r:[n] \rightarrow \mathbb{N}$. Then, for every  \emph{wFOEIL} sentence $\ps$ over $pw\B$ and $K$ we can effectively construct a \emph{WFA} $\mathcal{A}_{\ps,r}$ over $I_{p\B(r)}$ and $K$ such that $\Vert \ps \Vert (r,w)=\Vert\mathcal{A}_{\ps,r}\Vert(w)$ for every  $w \in I_{p\B(r)}^*$. The worst case run time for the translation algorithm is doubly exponential and the best case is exponential.    
\end{theorem}

We shall prove Theorem \ref{wsent_to_waut} using the subsequent Proposition \ref{wformula-waut}. For this, we need to slightly modify a corresponding result from \cite{Pi:Ar}. More precisely, we state the next proposition.
\begin{proposition}\label{form-aut}
Let $\psi$ be a \emph{FOEIL} formula over a set $p\B=\{B(i,j) \mid i\in [n], j\geq 1\}$ of parametric components. Let also $\mathcal{V} \subseteq \mathcal{X}$ be a finite set of variables containing $\mathrm{free}(\psi)$, $r:[n] \rightarrow \mathbb{N}$, and $\sigma$ a  $(\mathcal{V}, r)$-assignment. Then, we can effectively construct a finite automaton $\mathcal{A}_{\psi,r, \sigma}$ over $I_{p\B(r)}$ such that $(r, \sigma, w) \models \psi$ iff $w \in L(\mathcal{A}_{\psi,r, \sigma})$ for every $w \in I_{p\B(r)}^*$. The worst case run time for the translation algorithm is exponential and the best case is polynomial. 
\end{proposition}
\begin{proof}
We modify the proof of Proposition 23 in \cite{Pi:Ar} as follows.
If $\psi=\mathrm{true}$, then we consider the complete finite automaton $\mathcal{A}_{\psi,r, \sigma}=(\{q\}, I_{p\B(r)},q, \Delta, \{q\})$ with $\Delta=\{(q,a,q) \mid   a \in I_{p\B(r)} \}$. If $\psi=p(x^{(i)})$, then we consider the deterministic finite automaton  $\mathcal{A}_{\psi, r, \sigma}=(\{q_0,q_1\}, I_{p\B(r)}, q_0, \Delta, \{q_1\})$ with $\Delta=\{(q_0,a,q_1) \mid   p(\sigma(x^{(i)})) \in a\}$. Then, we follow accordingly the same induction steps. Concerning the complexity of the translation we use the same arguments (we do not take into account the trivial case $\psi = \mathrm{true}$ where the complexity of the translation is constant). 
\end{proof}

\begin{proposition}\label{wformula-waut}
Let $\ps$ be a \emph{wFOEIL} formula over a set $pw\B=\{wB(i,j) \mid i\in [n], j\geq 1\}$ of parametric weighted components and $K$. Let also $\mathcal{V} \subseteq \mathcal{X}$ be a finite set of variables containing $\mathrm{free}(\psi)$, $r:[n] \rightarrow \mathbb{N}$ and $\sigma$  a $(\mathcal{V}, r)$-assignment. Then, we can effectively construct a \emph{WFA} $\mathcal{A}_{\ps,r,\sigma}$ over $I_{p\B(r)}$ and $K$ such that $\Vert \ps \Vert(r, \sigma, w) = \Vert \mathcal{A}_{\ps,r, \sigma}\Vert(w)$ for every $w \in I_{p\B(r)}^*$. The worst case run time for the translation algorithm is doubly exponential and the best case is exponential.   
\end{proposition}
\begin{proof}
We prove our claim by  induction on the structure of the wFOEIL formula $\ps$. 
\begin{itemize}
\item[i)] If $\ps=k$, then we consider the WFA $\mathcal{A}_{\ps,r,\sigma}=(\{q\}, in, wt, ter)$ over $I_{p\B(r)}$ and $K$ with $in(q)=k$, $wt(q,a,q)=1$ for every $a \in I_{p\B(r)}$, and $ter(q)=1$. 
\item[ii)] If $\ps=\psi$, then we consider the finite automaton $\mathcal{A}_{\psi,r,\sigma}$ derived in Proposition \ref{form-aut}. Next, we construct, in exponential time, an equivalent complete finite automaton $\mathcal{A}'_{\psi,r,\sigma}=(Q,I_{p\B(r)},q_0,\Delta,F)$. Then, we construct, in linear time,  the WFA $\mathcal{A}_{\ps,r,\sigma}=(Q,in,wt,ter)$ where $in(q)=1$ if $q=q_0$ and $in(q)=0$ otherwise, for every $q \in Q$, $wt(q,a,q')=1$ if $(q,a,q') \in \Delta$ and $wt(q,a,q')=0$ otherwise, for every $(q,a,q') \in Q\times A \times Q$, and $ter(q)=1$ if $q \in F$ and $ter(q)=0$ otherwise, for every $q\in Q$.
\item[iii)] If $\ps=\ps_1 \oplus \ps_2$ or $\ps=\ps_1 \otimes \ps_2$ or $\ps=\ps_1 \odot \ps_2$ or $\ps=\ps_1 \varpi \ps_2$, then we rename firstly variables in $\mathrm{free}(\ps_1) \cap \mathrm{free}(\ps_2)$ as well variables which are free in $\ps_1$ (resp. $\ps_2$) and bounded (i.e., not free) in $\ps_2$ (resp. in $\ps_1$). Then, we extend $\sigma$ on $\mathrm{free}(\ps_1) \cup \mathrm{free}(\ps_2)$ in the obvious way, and construct   $\mathcal{A}_{\ps,r,\sigma}$ from   $\mathcal{A}_{\ps_1,r,\sigma}$ and $\mathcal{A}_{\ps_2,r,\sigma}$  by applying Proposition \ref{prop-rec}.
\item[iv)] If $\ps={\textstyle\sum x^{(i)}}. \ps' $, then we get $\mathcal{A}_{\ps,r,\sigma}$ as the WFA for the sum of the series $\Vert\mathcal{A}_{\ps',r,\sigma[x^{(i)}\rightarrow j]}\Vert$, $j\in [r(i)]$ (Proposition \ref{prop-rec}).
\item[v)] If $\psi={\textstyle\prod x^{(i)}}. \ps'$, then we get $\mathcal{A}_{\ps,r,\sigma}$ as the WFA for the Hadamard product of the series $\Vert\mathcal{A}_{\ps',r,\sigma[x^{(i)}\rightarrow j]}\Vert$, $j\in [r(i)]$ (Proposition \ref{prop-rec}).
\item[vi)] If $\ps={\textstyle\sum ^{\odot} x^{(i)}}. \ps' $, then we compute firstly all nonempty  subsets $J$ of $[r(i)]$. For every such subset $J=\{l_1, \ldots, l_t\}$, with $1 \leq l_1 <  \ldots <  l_k \leq r(i)$, we consider the WFA $\mathcal{A}_{\ps,r,\sigma}^{(J)}$ accepting the Cauchy product of the series $\Vert\mathcal{A}_{\ps',r,\sigma[x^{(i)}\rightarrow l_1]}\Vert,\ldots, \Vert\mathcal{A}_{\ps',r,\sigma[x^{(i)} \rightarrow l_t]}\Vert$. Then, we get $\mathcal{A}_{\ps,r,\sigma}$ as the WFA for the sum of the series  $\Vert\mathcal{A}_{\ps,r,\sigma}^{(J)}\Vert$  with $\emptyset \neq J \subseteq [r(i)]$ (Proposition \ref{prop-rec}). 
\item[vii)] If $\ps={\textstyle\prod ^{\odot} x^{(i)}}. \ps' $, then we get $\mathcal{A}_{\ps,r,\sigma}$ as the WFA for the Cauchy product of the series $\Vert\mathcal{A}_{\ps',r,\sigma[x^{(i)}\rightarrow j]}\Vert$, $j\in [r(i)]$ (Proposition \ref{prop-rec}).
\item[viii)] If  $\ps={\textstyle\sum ^{\varpi} x^{(i)}}. \ps' $, then we compute firstly all nonempty  subsets $J$ of $[r(i)]$. For every such subset $J=\{l_1, \ldots, l_t\}$, with $1 \leq l_1 <  \ldots <  l_k \leq r(i)$, we consider the WFA $\mathcal{A}_{\ps,r,\sigma}^{(J)}$ accepting the shuffle product of the series $\Vert\mathcal{A}_{\ps',r,\sigma[x^{(i)}\rightarrow l_1]}\Vert,\ldots, \Vert\mathcal{A}_{\ps',r,\sigma[x^{(i)} \rightarrow l_t]}\Vert$. Then, we get $\mathcal{A}_{\ps,r,\sigma}$ as the WFA for the sum of the series  $\Vert\mathcal{A}_{\ps,r,\sigma}^{(J)}\Vert$  with $\emptyset \neq J \subseteq [r(i)]$ (Proposition \ref{prop-rec}). 
\item[ix)] If $\ps={\textstyle\prod ^{\varpi} x^{(i)}}. \ps' $, then we get $\mathcal{A}_{\ps,r,\sigma}$ as the WFA for the shuffle product of the series $\Vert\mathcal{A}_{\ps',r,\sigma[x^{(i)}\rightarrow j]}\Vert$, $j\in [r(i)]$ (Proposition \ref{prop-rec}).
\end{itemize}
By our constructions above, we immediately get  $\Vert \ps\Vert(r, \sigma, w) = \Vert \mathcal{A}_{\ps,r,\sigma} \Vert(w)$ for every $w \in I_{p\B(r)}^*$ . Hence, it remains to deal with the time complexity of our translation algorithm. 
  
Taking into account the above induction steps, we show that the worst case run time for our translation algorithm is doubly exponential. Indeed, if $\ps'=\psi$ is a FOEIL formula, then our claim holds by (ii) and Proposition \ref{form-aut}. Then the constructions in steps (iii)-(v), (vii) and (ix) require a polynomial time (cf. Proposition \ref{prop-rec}). Finally, the translations in steps (vi) and (viii) require at most a doubly exponential run time because of the following reasons. Firstly, we need to compute all nonempty subsets of $[r(i)]$ which requires an exponential time. Then, due to our restrictions for $\ps'$ in $\ps={\textstyle\sum ^{\odot} x^{(i)}}. \ps' $ and $\ps={\textstyle\sum ^{\varpi} x^{(i)}}. \ps' $, and Proposition \ref{form-aut} (cf. also the proof of Proposition 24 in \cite{Pi:Ar}), if a FOEIL subformula $\psi$ occurs in $\ps'$, then we need a polynomial time to translate it to a finite automaton  and by (ii) an exponential time to translate it to a WFA. We should note that if $\ps'$ contains a subformula of the form $\exists ^* x^{(i')}.\psi''$ or $\exists ^{\shuffle} x^{(i')}.\psi''$ or ${\textstyle\sum ^{\odot} x^{(i')}}. \ps'' $ or ${\textstyle\sum ^{\varpi} x^{(i')}}. \ps'' $, then the computation of the subsets of $r[i']$ is independent of the computation of the subsets of $r[i]$. 
On the other hand, the best case run time of the algorithm is exponential. Indeed, if in step (ii) we get $\mathcal{A}_{\psi, r, \sigma}$ in polynomial time (cf. Proposition \ref{form-aut}) and we need  no translations of steps (vi) and (viii), then the required time is exponential.
\end{proof}

\

Now we are ready to state the proof of Theorem \ref{wsent_to_waut}.

\begin{proof}[Proof of Theorem \ref{wsent_to_waut}] We apply Proposition \ref{wformula-waut}. Since $\ps$ is a weighted sentence it contains no free variables. Hence, we get a WFA $\mathcal{A}_{\ps,r}$ over $I_{p\B(r)}$ and $K$ such that $\Vert\ps\Vert(r,w)=\Vert\mathcal{A}_{\ps,r}\Vert(w)$ for every $w \in I^*_{p\B(r)}$, and this concludes our proof. The worst case run time for the translation algorithm is doubly exponential and the best case   is exponential.
\end{proof}

\

Next we prove the decidability of the equivalence of wFOEIL sentences over (subsemirings of) skew fields. It is worth noting that the complexity remains the same with the one for the decidability of equivalence for FOEIL formulas \cite{Pi:Ar}.

\begin{theorem}
Let $K$ be  a (subsemiring of a) skew field, $pw\mathcal{B}= \{wB(i,j) \mid i \in [n], j \geq 1 \} $  a set of parametric weighted components over   $K$,  and  $r:[n] \rightarrow \mathbb{N}$ a mapping. Then, the equivalence  problem for \emph{wFOEIL} sentences over $pw\B$ and $K$ w.r.t. $r$ is decidable in doubly exponential time. 
\end{theorem} 
\begin{proof}
It is well known that the equivalence problem for weighted automata, with weights taken in  (a subsemiring of) a skew field, is decidable in cubic time (cf. Theorem 4.2 in \cite{Sa:El}, \cite{Sa:Ra}). Hence, we conclude our result by Theorem \ref{wsent_to_waut}.
\end{proof}

\begin{corollary}
Let $pw\mathcal{B}= \{wB(i,j) \mid i \in [n], j \geq 1 \} $ be a set of parametric weighted components over $\mathbb{Q}$ and  $r:[n] \rightarrow \mathbb{N}$ a mapping. Then, the equivalence  problem for \emph{wFOEIL} sentences over $pw\B$ and $\mathbb{Q}$ w.r.t. $r$ is decidable in doubly exponential time. 

\end{corollary}

\section{Conclusion}
In this paper we studied the modelling of architectures for 
parametric weighted component-based systems. We introduced a weighted first-order extended interaction logic, wFOEIL, over a finite set of ports and a commutative semiring to
characterize quantitative aspects of architectures, such as the minimum
cost or the maximum probability of the implementation of interactions, depending on the underlying semiring. Our wFOEIL models parametric weighted interactions by preserving their execution order as imposed by the corresponding architecture.
Moreover, we showed that the equivalence problem for wFOEIL sentences over (a subsemiring of) a skew field is decidable in doubly exponential time. Furthermore, we applied wFOEIL for describing well-known parametric architectures in the quantitative setting. 

Work in progress involves the investigation of the second-order level of our wEPIL 
over semirings, in order to capture the quantitative aspects of more complicated parametric architectures such as Ring, Pipeline, or  Grid (cf. \cite{Bo:St, Ma:Co}). Future research also includes the study of our weighted logics over alternative weight structures, found in applications, like for
instance valuation monoids \cite{Dr:Re,Ka:We}. Existing work has extensively studied the verification problem for parametric qualitative systems with specific topology or communication rules
against invariant properties \cite{Bo:Ch,Bo:St,Pn:Au,Zu:In} or temporal properties \cite{Am:Pa,Bl:De, Ch:On}. On the other hand, a few works have investigated verification problems of probabilistic parametric systems \cite{An:Pa, Be:Pa, Es:Po}. 
Therefore, it would be very interesting to study verification problems of parametric weighted systems
within our modelling architecture framework. Finally, another research direction is the 
implementation of our results in component-based frameworks
for automating the modelling and architecture identification of arbitrary parametric weighted systems.

\

\hide{\noindent \textbf{Acknowledgement.} We are deeply grateful to Simon Bliudze for discussions on a previous version of the paper. }

\end{document}